\documentclass[journal]{IEEEtran}
\ifCLASSINFOpdf
\else
\fi
\usepackage{subfigure} 
\usepackage{graphicx}
\usepackage{amssymb}
\usepackage{amsmath}
\usepackage{color} 
\usepackage{cite} 
\newtheorem{theorem}{Theorem}
\newtheorem{corollary}{Corollary}
\newtheorem{lemma}{Lemma}
\newtheorem{definition}{Definition}
\newtheorem{assumption}{Assumption}
\newtheorem{remark}{Remark}
\newtheorem{proof}{Proof}
\newtheorem{problem}{Problem}
\newtheorem{proposition}{Proposition}
\usepackage{tikz,xcolor,hyperref}
\definecolor{lime}{HTML}{A6CE39}
\DeclareRobustCommand{\orcidicon}{%
	\begin{tikzpicture}
		\draw[lime, fill=lime] (0,0) 
		circle [radius=0.16] 
		node[white] {{\fontfamily{qag}\selectfont \tiny ID}}; 
		\draw[white, fill=white] (-0.0625,0.095) 
		circle [radius=0.007];	  
	\end{tikzpicture}
	\hspace{-2mm}}
\foreach \x in {A, ..., Z}{%
	\expandafter\xdef\csname orcid\x\endcsname{\noexpand\href{https://orcid.org/\csname orcidauthor\x\endcsname}{\noexpand\orcidicon}}
}
\makeatother

\begin{document}
	\title{{Distributed Power Allocation Scheme with Prescribed Performance and Intermittent Dynamics for BESSs with Discharge Rate Constraints in Microgrids}}
	\author{Yalin Zhang\orcidA{},
		Zhongxin Liu\orcidB{},~\IEEEmembership{Member,~IEEE,}
		Fuyong Wang\orcidC{},
		and Zengqiang Chen\orcidD{}
		\thanks{Manuscript received XX, XX; revised XX, XX;
			accepted XX, XX. This work is supported by the National Natural Science Foundation of China (Grant Nos. 62103203 and 92367105) and the General Terminal IC Interdisciplinary Science Center of Nankai University (Corresponding author:  Fuyong Wang.). 
			\par The authors are with the College of Artificial Intelligence, Nankai University, Tianjin 300350, and also with the Tianjin Key Laboratory of Interventional Brain-Computer Interface and Intelligent Rehabilitation, Nankai University, Tianjin 300350, China (e-mail: zhangyl@mail.nank.edu.cn; lzhx@nankai.edu.cn; wangfy@nankai.edu.cn; chenzq@nankai.edu.cn).
		}
	}
	\markboth{IEEE TRANSACTIONS ON SYSTEMS, MAN, AND CYBERNETICS: SYSTEMS, VOL. XX, NO. XX, XX XX}%
	{Zhang \MakeLowercase{\textit{et al.}}: \textcolor{blue}{Distributed Power Allocation Scheme with Prescribed Performance and Intermittent Dynamics for BESSs with Discharge Rate Constraints in Microgrids}}
	\maketitle
	\begin{abstract}
		The State-of-Charge (SoC) is an important parameter of a battery energy storage system (BESS), and its balance problem is also an issue worth studying in a multi-BESS network. Recently, some researchers have proposed a power allocation method, claiming that as long as the power sharing state and SoC balance state can be obtained in real-time, it can not only maintain supply and demand balance, but also ensure that any BESS will not exit early due to insufficient energy storage. Considering this, we are attempting to design {a distributed dynamic average tracking algorithm} based on multi-agent systems (MASs) with prescribed transient and steady state performance to estimate the power sharing and SoC balance states \textcolor{blue}{of} a BESSs network with dynamic load in a distributed manner. Two versions of the solution are designed here, where one is event triggered and the other is self triggered. These two schemes ensure the performance of the estimators under intermittent communication to varying degrees, \textcolor{blue}{respectively}. In each constructed estimation scheme, prescribed performance control (PPC) method is implemented to ensure the expected steady-state and dynamic performance, which is encapsulated in a performance function. Thus, it achieves almost zero error estimation of the fast time-varying {power} sharing and the SoC balance states. In addition, consensus and average tracking performance are decoupled, which provides convenience for \textcolor{blue}{parameters} tuning. Finally, to verify the results of \textcolor{blue}{the} theoretical analysis, some cases are studied and relevant simulations are performed on a 4 bus system.
	\end{abstract}
	\IEEEpeerreviewmaketitle
	\begin{center}
		N\footnotesize{OMENCLATURE}
	\end{center}
	\normalsize
	\begin{tabbing}
		\hspace{2cm} \= \kill
		BESS\> battery energy storage system\\
		SoC\>  the State-of-Charge\\
		PPC\> prescribed performance control\\
		MAS\> multi-agent system\\
		power\> proportional output power\\
		{DAT}\> {dynamic average tracking}\\
		$E_i$, $I_i$\> SoC and output current\\
		$\mathrm{C}_i$, $V_i$\> the capacity and terminal voltage\\
		$\eta_{i,1}$\>the Coulomb efficiency\\
		$\eta_{i,2}$\>the charging/discharging efficiency\\
		$\mathrm{C}_i$, \textcolor{blue}{$\mathrm{C}_{i,m}$}\>the current and \textcolor{blue}{the} rated capacities\\
		$\mathrm{T}$\>the time scale\\
		\textcolor{blue}{$f_i$}\>the current number of cycles\\
		$\tilde P_i$, $P_i$\> output power and its proportional value\\
		$\mathrm{K}_i^E$\> a constant defined as \textcolor{blue}{$\mathrm{K}_i^E=\frac{\eta_{i,1}}{\eta_{i,2}\mathrm{C}_iV_i}$}\\
		$D_i$\> the proportional load\\
		$P_a$, $E_a$\> the average power and \textcolor{blue}{the average} SoC\\
		$\hat P_{a,i}$\> the estimated average power\\
		$\hat E_{a,i}$\> \textcolor{blue}{the estimated average} SoC \\
		$\epsilon_E$, $\epsilon_P$\> the pre-specified values\\
		$\cal G$\> communication graph\\
		$\cal V$, $\cal E$\> vertex set and edge sets\\
		$\cal A$\> the adjacency matrix of $\cal G$\\
		${\cal N}_i$\> the neighbor set \textcolor{blue}{agent $i$}\\
		{$\rho^E_{ij}$}\>{the performance function of SoC balance error}\\
		{\textcolor{blue}{$\rho_{ij}^P$}, $\rho_l$, $R$}\>the performance functions\\
		$\rho_0$, $\rho_{\infty}$, $\lambda$\> the initial, final values and decay rate of $\rho(t)$\\
		$\xi^P$, $\xi^E$\> modulated error vector forms\\
		$T^P$, $T^E$\> bijective mapping vector forms \textcolor{blue}{of $\xi^P$ and $\xi^E$}\\
		$J_T^P$, $J_T^E$\> the Jacobian derivative of $T^P$ \textcolor{blue}{and} $T^E$\\
		$\bar{\epsilon}$, $K$, $F$\> \textcolor{blue}{several constants} defined in Lemma \ref{lepre}\\
		$\epsilon^*$, $\xi^*$ \> \textcolor{blue}{several constants} defined in the proof\\
		$\mathrm{a}_1$, $\mathrm{a}_2$\> maximum and minimum allowed SoC values\\
		$z_i^P$, $z_i^E$\> the intermediate states of the two estimators\\
		$\delta_l^P$\> consensus errors for estimation\\
		$\delta^P$\> consensus error vectors of \textcolor{blue}{$\delta_l^P$}\\
		$\Omega_\xi$, $\Omega_\xi^{'}$\> an open and nonempty set and its subset\\
		$V$, $\dot V$\> a Lyapunov function and its derivative\\
		$e^P$\> a tracking error vector\\
		$\mathrm{m}$, $\mathrm{n}$\> numbers of communication links and agents\\
		$l$, $i$\> the index of communication links and agents\\
		{$\mathrm{k}^P$, $\mathrm{k}^P_r$}\> {the gains of the power sharing state estimator}\\
		{$\mathrm{k}^E$, $\mathrm{k}^E_r$}\> {the gains of the SoC balance state estimator}\\
	\end{tabbing}
	\section{Introduction}
	\IEEEPARstart{A}{s} the main source of energy for centralized power generation, fossil fuels such as coal, used for power generation in power plants, are extensively mined and used, causing global pollution and greenhouse effect. Nowadays, the difficulty of extracting traditional energy has increased, coupled with the damage to the environment, making it necessary for people to make new choices. Thus, more types and proportions of renewable energy in distributed manner are utilized in \textcolor{blue}{microgrids} \cite{HossainLipu2022, ZHAO2014538, solyaliComprehensiveStateoftheartReview2022}. In practical application, renewable energy generation has the characteristics of randomness and intermittence, which brings that microgrids often cannot provide reliable and stable power \cite{HossainLipu2022,ZHAO2014538}, {as shown in Fig. \ref{fige}}. Recently, energy storage system is usually introduced into microgrid,
	which is helpful to improve power supply quality and suppress peak valley differences \cite{ZHAO2014538, solyaliComprehensiveStateoftheartReview2022}. In all kinds of energy storage devices, battery energy storage system (BESS) has strong advantages in usage scenario, response speed, quality and efficiency \cite{Yang2022,Shinichi2017,Sokol2017}. As a result, an increasing proportion of BESSs are integrated into the smart grid. Herein above statement, cooperative management/control of BESSs becomes a problem worthy of in-depth study.
	\begin{figure}
		\centering
		\includegraphics[width=8cm]{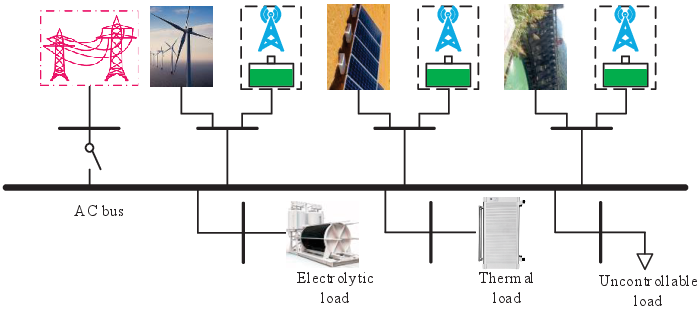} \caption{{Architecture of isolated MGs with BESSs.}\label{fige}}
	\end{figure}
	\par For a battery cell, the State of Charge (SoC) is an importantparameter that measures its current energy level \cite{Zeng2022l,Wang2020d,Litjens2018}. For parallel battery packs, SoCs \textcolor{blue}{of} all battery cells need to be balanced to maximize battery capacity utilization \cite{Zeng2022l,Litjens2018}. In addition, achieving SoC balance often extends the lifespan of battery packs \cite{Wang2020d,Litjens2018}. Therefore, how to reasonably arrange power flow to achieve SoC balance is a fundamental issue for achieving the cooperative control/management objectives of BESSs.
	\par Thus, a large number of excellent literature has been committed to solving the above problems. Both distributed and centralized manners have been investigated. For example, SoC of each battery was monitored by a centralized controller in \cite{Baveja2023}, while external balancing circuit was applied to realize SoC balance. Besides, a centralized hierarchical SoC balancing controller was developed in \cite{Cao2021} to achieve current sharing control and SoC balancing control. However, the external balancing circuit was going to consume some energy. In practice, \textcolor{blue}{a} centralized controller was more expensive, and single-point failures are unavoidable. Instead, distributed control could enhance the robustness and scalability of the system.
	\par In view of this, the research on distributed control \textcolor{blue}{schemes} for BESSs based on multi-agent systems (MASs) has attracted great attention and achieved some excellent results. SoC balance \cite{Cai2016,Khazaei2019a,Ding2020,Hu2019,Nguyen2021a,Zhang2020} and power/current sharing \cite{Khazaei2019a,Ding2020,Hu2019,Nguyen2021a,Zhang2020} with constant load were studied in distributed control frameworks respectively. The relevant authors designed proportional control protocols \cite{Cai2016,Khazaei2019a,Nguyen2021a,Zhang2020}, finite time/fixed time control protocols \cite{Ding2020,Hu2019} and resilient control protocol \cite{Ding2020} in these papers in order to accelerate system convergence and improve anti-interference ability. Besides, other methods, e.g., model-free control \cite{Hong2021d} and sliding mode control \cite{Zhang2019e}, were also applied to reach SoC balance. These methods were used to achieve SoC consensus by proportionally distributing a constant load to solve SoC balancing problem and demonstrate good performance. {However, these solutions are inadequate in solving the SoC balancing problem of BESS networks with dynamic loads, exhibiting defects such as poor control accuracy and supply-demand imbalance. In addition, these solutions come at the cost of changing some BESS working modes in the network, which, although achieving SoC balance, prompt these BESSs to advance to the next work cycle. All of these forces us to take a different path.}
	\par Recently, some latest developments have emerged in \cite{Xing2019,Meng2021,Meng2022,Qian2023,Wu2022} {to solve the SoC balancing problem under time-varying loads}. The authors believed in \cite{Xing2019} that all BESSs should {output power according to the current SoC, and based on this, provided a power allocation rule. To ensure supply-demand balance and simultaneous depletion of all BESSs, it is necessary to design two estimators to obtain real-time power sharing status and SoC balance status}. \textcolor{blue}{Considering} this, the authors in \cite{Meng2021} designed two {dynamic average tracking (DAT)} schemes {to achieve this goal, where} one was {a asymptotic scheme with controllable steady-state error}, and the other {was a} finite time scheme. Similarly, the authors further improved this method in \cite{Meng2022} and used adaptive methods to estimate battery parameter. Furthermore, in \cite{Qian2023}, the authors designed a scheme with a node-based event triggering mechanism to estimate the state of the battery, thereby saving communication resources. Additionally, a distributed SoC balance scheme based on discrete dynamics was designed in \cite{Wu2022} to reach relative SoC variation rates consensus. These results {are} found to solve the SoC balancing problem of {BESSs with a unified discharge rate under dynamic load}. And these authors either focused on the consensus performance \cite{Xing2019,Wu2022} of the designed algorithm or on the traking performance \cite{Meng2021,Meng2022,Qian2023}, which was because these two types of performance were often coupled, and adjusting the gain made one better and the other worse. {In addition, prior knowledge of dynamic loads was required in \cite{Meng2021,Meng2022,Qian2023} and could affect estimation accuracy, which might require a suitable load prediction mechanism to complete.}
	\par\textcolor{blue}{Since some researchers proposed the prescribed performance control (PPC) method \cite{Bechlioulis2008, Bechlioulis2009, Bechlioulis2010}, its application research in collaborative control of MASs has been gradually expanded. At present, PPC is applied in MASs to solve problems such as consensus \cite{10101860,10428061,10578027,10398255,10711307}/binary consensus \cite{10418189} control, formation \cite{10198705} and tracking control \cite{10418189}, fault-tolerant control \cite{10486878,10711307}, etc. These schemes combine finite \cite{10198705}/fixed \cite{10486878,10418189,10398255} time control, event triggered control \cite{10428061,10578027,10398255}, adaptive control \cite{10418189,10398255}, and reinforcement learning \cite{10486878}, etc, to accelerate convergence rate, save communication resources, approximate nonlinearity, and improve robustness. These results fully demonstrate that PPC has shone in distributed collaborative control of MASs, promoting the evolution of control errors within predefined ranges. There are also some fragility-free improvements \cite{9930639,BJF2022} made in PPC schemes to address actuator saturation and interference. However, for distributed DAT control schemes used to solve SoC balancing and power sharing problems, existing results have not been introduced into PPC to improve system performance.}
	\par {In order to design a DAT algorithm that decouples consensus and tracking performance to solve power allocation and SoC balancing problems under dynamic loads, we are going to design two distributed estimators with PPC and the intermittent execution\textbackslash communication to obtain the real-time SoC balance and power sharing states of a multi-BESSs network with power allocated according to the current SoC level}. 
	Concisely, the main contributions in this paper are listed as follows.
	\begin{enumerate}	
		\item \textcolor{blue}{The PPC method is introduced into a distributed DAT scheme to estimate the SoC balance state and the power sharing state, which makes that the estimation error of each agent evolves predetermined and arbitrary convergence rates}. Unlike \cite{Meng2021,Meng2022,Qian2023,Wu2022}, this is not achieved by adjusting the control gains, but through a prescribed performance function.
		\item The designed scheme does not require prior knowledge of the dynamic load to achieve average tracking. In addition, the designed scheme decouples consensus performance and tracking performance to select the appropriate gains. These are different from \cite{Meng2021,Meng2022,Qian2023,Wu2022}.
		\item An event and a self triggering mechanisms \textcolor{blue}{are designed}. {The designed event triggering mechanism allows any pair of neighboring agents to trigger asynchronously which is different from \cite{9770469}.} Using a conservative upper limit of measurement error as a substitute to construct a self triggering mechanism to achieve intermittent communication.
	\end{enumerate}
	\par The following contents are arranged as follows. \textcolor{blue}{In} Section II, a general description of SoC, a discharge rate constraint of BESSs and the common control objectives \textcolor{blue}{are introduced}. \textcolor{blue}{Next, in} Section III, the preliminaries of the graph theory, PPC and dynamical systems \textcolor{blue}{are  given}. \textcolor{blue}{Subsequently,} two distributed estimators and their proof for power sharing and balance states \textcolor{blue}{are designed in Section IV}. \textcolor{blue}{In addition,} several cases \textcolor{blue}{are presented in Section VI} to test the proposed estimators. \textcolor{blue}{Finally, this paper is summarized in} Section VII.
	\section{Problem statement}
	In this section, a power allocation scheme with a unified SoC relative rate of variation is introduced, which can ensure \textcolor{blue}{time-oriented} SoC balance and supply-demand balance. Based on this, the purpose of this article is to design two \textcolor{blue}{distributed} estimators that can ensure predetermined dynamic and steady-state performance.
	{\subsection{SoC of A BESS}}
	\par For a BESS, the current energy level of {the} battery {can be characterized by} SoC, and its variation rate is related to the output current of the battery, i.e., for BESS $i$ in a cycle,
	$$\dot{E}_i=-\frac{\eta_{i,1}}{\eta_{i,2}\mathrm{C}_i\mathrm{T}}I_i,$$
	where $E_i$, $\eta_{i,1}$, $\eta_{i,2}$, $\mathrm{C}_i$, and $\textcolor{blue}{I_i}$ {are}, respectively, SoC,
	Coulomb efficiency, charging/discharging efficiency, capacity, and output current of BESS $i$, $\eta_{i,1}$ \textcolor{blue}{and} $\eta_{i,2}$ are constants, $\mathrm{T}$ denotes \textcolor{blue}{a} time scale. Considering that the output power of the battery unite of BESS $i$ is calculated by $\tilde{P}_{i}=V_{i}{I}_{i}$, where $V_i$ is the terminal voltage of battery $i$, the rate-of-change of $E_i$ is related to the output power of the battery \cite{Cai2016,Khazaei2019a,Ding2020,Hu2019,Nguyen2021a,Zhang2020}, \cite{Zhang2019e,Xing2019,Meng2021,Meng2022,Wu2022}. That is, for \textcolor{blue}{BESS} $i$ and $i\in\{1,2,\cdots,\mathrm{n}\}$,
	$$\dot{E}_i=-\frac{\mathrm{K}_i^E}{\mathrm{T}}\tilde{P}_i,$$
	where $\mathrm{K}_i^E=\frac{\eta_{i,1}}{\eta_{i,2}\mathrm{C}_iV_i}$ is considered a parameter that characterizes the capacity of BESS $i$, {$V_i$ can be kept constant by adjusting the duty cycle of the DC chopper. In view of this,} \textcolor{blue}{one can get}
	\begin{equation}
		\label{1}
		\dot{E}_i=-\frac{1}{\mathrm{T}}P_i,
	\end{equation}
	where $P_{i}=\mathrm{K}_{i}^{\hat{E}}\tilde{P}_{i}$ is \textcolor{blue}{the} proportional output power of \textcolor{blue}{BESS} $i$, \textcolor{blue}{which is} so-called output power and hereinafter referred to as power. Here, it is believed that SoC satisfies the following assumption.
	\begin{assumption}
		To ensure the safety of the battery and extend its lifespan, SoC of BESS $i$ for $i\in\{1,2,\cdots,{\mathrm{n}}\}$ is often bounded by
		$$\mathrm{a}_1\leq E_i\leq\mathrm{a}_2,$$
		where $\mathrm{a}_1$ and $\mathrm{a}_2$ are positive constants and $0<\mathrm{a}_1<\mathrm{a}_2<\eta_{i,1}^f${, and $\eta_{i,1}^f$ is the current number of cycles of \textcolor{blue}{BESS} $i$}.
	\end{assumption}
	\par {The actual capacity of \textcolor{blue}{BESS} $i$ after $f_i$ cycles is
		$$\mathrm{C}_i=\mathrm{C}_{i,m}\eta_{i,1}^{f_i},$$	
		where $\mathrm{C}_{i,m}$ is the rated capacity of \textcolor{blue}{BESS} $i$. It can be seen from here that as the number of cycles increases, the battery capacity continuously decreases.}
	{\subsection{A Power Allocation Scheme for BESSs}}
	\par {For droop-controlled BESSs, the previous power allocation schemes are achieved based on proportional power/marginal cost consensus, which, essentially, either lead to circulation and changes in operating mode of certain batteries, which further cause capacity degradation according to the above subsection, or SoC balancing cannot be achieved. That is}, whether the battery has \textcolor{blue}{a} control input or is regulated through droop control, existing solutions cannot avoid the cycle of charging and discharging of certain batteries during the process of achieving SoC balance. This is actually not conducive to the protection of battery life. {Recently}, a SoC balance scheme is developed to ensure that all batteries in a BESS network are in the same mode \cite{Xing2019,Meng2021,Meng2022,Qian2023,Wu2022}, i.e.,
	\begin{equation}
		\label{2}
		\frac{P_i}{E_i}=\frac{P_j}{E_j},\forall i,j\in\{1,2,\cdots,\mathrm{n}\}.
	\end{equation}
	After a simple calculation {according to Equal Ratios Theorem}, we arrive at {$\frac{P_i}{E_i}=\frac{\sum_{j=1}^{n}P_j}{\sum_{j=1}^{n}E_j}$, thereby}
	\begin{equation}
		\label{3}
		{P_i=\frac{P_\mathrm{a}}{E_\mathrm{a}}E_i,}
	\end{equation}
	{where $P_\mathrm{a}=\frac{1}{n}\sum_{j=1}^{n}\textcolor{blue}{P}_j$, $E_\mathrm{a}=\frac{1}{n}\sum_{j=1}^{n}E_j$.} {Here, we cite a conclusion from previous literature to construct a theorem to illustrate the effect of this mechanism on BESSs.}
	\begin{theorem}
		\label{Th1}
		\cite{Xing2019,Meng2021,Meng2022,Qian2023,Wu2022} For a resistive network with $\mathrm{n}$ BESSs, they have the same SoC relative rate of change as shown in {\eqref{2}}. In this way, SoC balance can be guaranteed among all batteries, while the supply-demand balance can be maintained.
	\end{theorem}
	\par According to Theorem \ref{Th1}, as long as the average power and SoC can be obtained in real-time, then all battery energy will be depleted simultaneously {under the power allocation scheme \eqref{3}}, which clearly is a centralized approach.
	\par With the help of distributed {technology}, some distributed {DAT} algorithms to estimate $P_a$ and $E_a$ are designed by authors in \cite{Xing2019,Meng2021,Meng2022,Qian2023,Wu2022}. So, a power allocation algorithm that facilitates the design of distributed solutions is designed as follows,
	\begin{equation}
		\label{Pcal}
		P_i=\frac{E_i}{\max\{\mathrm{a}_1,\hat{E}_{\mathrm{a},i}\}}\hat{P}_{\mathrm{a},i},
	\end{equation}
	where ${\hat E}_{a,i}$ and ${\hat P}_{a,i}$ are the real-time estimated values of the average SoC and the average output power. So, Theorem \ref{Th1} can be extended to the following proposition.
	\begin{proposition}
		\label{pro1}
		For a resistive network with $\mathrm{n}$ BESSs, {the output power of} each one follows \eqref{Pcal}. Provided that the average values of load and SoCs can be well estimated, time-oriented SoC balance can be guaranteed among all \textcolor{blue}{BESSs}, while the supply-demand balance can be maintained.
	\end{proposition}
	\subsection{Control Objectives}
	\par Based on Proposition \ref{pro1}, many researchers are committed to developing leader-follower-based \cite{Xing2019,Meng2021,Meng2022}, PI-based \cite{Xing2019,Meng2021,Meng2022}, finite-time \cite{Meng2021}, and event-triggered-based \cite{Qian2023} distributed solutions. \textcolor{blue}{It is noted} that existing solutions require prior knowledge of dynamic loads, and consensus performance and tracking performance are coupled, which makes it difficult to choose gains. To alleviate this dilemma, the PPC method is introduced to constrain consensus performance and design a mechanism with intermittent execution/communication. Provide a detailed description of the objectives of this article in Problem 1.
	\begin{figure}
		\centering
		\includegraphics[width=8cm]{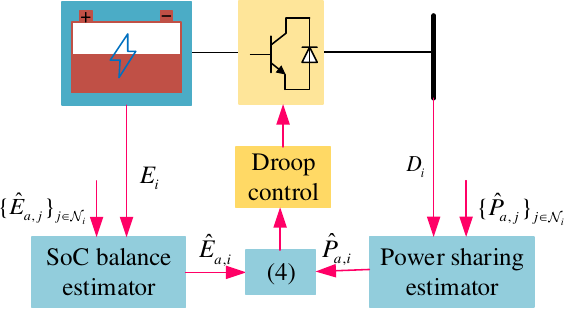} \caption{{A power distribution scheme for droop-controlled BESSs.}\label{figo}}
	\end{figure}
	\begin{problem} \label{problem}
		Consider a resistive network with $\mathrm{n}$ {droop-controlled} BESSs, each of which is equipped with \eqref{Pcal}. Two distributed estimators need to be designed to address the following issues.	
		\\(1) The {power} sharing state is well estimated so that the convergence rate is not less than the predetermined value, and the error does not exceed the pre-specified value in steady-state, i.e.,
		$$\lim\limits_{t\to+\infty}\left| {{\hat P_{a,i} } - {D_a}} \right|\le \epsilon_P,$$
		where $i \in \left\{ {1,2, \cdots ,N} \right\}$, ${D_a}=\frac{1}{n}\sum\limits_{i=1}^{n}D_i$, $D_i$ is the load \textcolor{blue}{power} that BESS $i$ needs to share, $\epsilon_E$ is a pre-specified value.\\
		(2) The SoC balance state is well estimated so that the convergence rate is not less than the predetermined value, and the error does not exceed the pre-specified value in steady-state, i.e., 
		$$\lim\limits_{t\to+\infty}\left| {{\hat E_{a,i}} - {E_{a}}} \right| \le \epsilon_E,$$
		where $i \in \left\{ {1,2, \cdots ,N} \right\}$, $\epsilon_E$ is a pre-specified value.
	\end{problem}
	\section{Preliminaries}
	{Some useful knowledge is introduced in this section.} Algebraic graph theory is introduced here as a useful tool to \textcolor{blue}{describe} communication technologies between agents. Some relevant knowledge about PPC and dynamical systems is also introduced in this section for estimators design and stability proof.
	\subsection{Graph Theory}
	\par The method proposed in this paper is based on MASs to design distributed estimators. Communication is required among agents managing BESSs to achieve control objectives. For a communication topology ${\cal G} = ({\cal V},{\cal E})$ for agents managing BESSs, it often includes some vertices, i.e., $\textcolor{blue}{\mathrm{v}}_i$ for $i\in\{1,2,\cdots,\textcolor{blue}{\mathrm{n}}\}$ and ${\cal V}=\{\textcolor{blue}{\mathrm{v}}_i|i\in \{1,2,\cdots,\textcolor{blue}{\mathrm{n}}\}\}$, and edges with these vertices as endpoints, i.e., $(\textcolor{blue}{\mathrm{v}}_i,\textcolor{blue}{\mathrm{v}}_j)$ and $\mathcal{E}=\{(\textcolor{blue}{\mathrm{v}}_i,\textcolor{blue}{\mathrm{v}}_j)|\textcolor{blue}{\mathrm{v}}_i,\textcolor{blue}{\mathrm{v}}_j\in\mathcal{V}\}$. As a result, vertexes correspond to each agent one by one. If these edges are all directed, the graph is called {a directed one}; Otherwise, it is called {a undirected one}. In {a} directed graph version, the starting node $\textcolor{blue}{\mathrm{v}}_i$ of a directed edge, such as $(\textcolor{blue}{\mathrm{v}}_i,\textcolor{blue}{\mathrm{v}}_j)$, is considered as the neighbor of the end node $\textcolor{blue}{\mathrm{v}}_j$. In other words, the information of the agent $i$ is considered to flow to and usable by \textcolor{blue}{agent $j$}. In \textcolor{blue}{an} undirected graph version, $(\textcolor{blue}{\mathrm{v}}_i,\textcolor{blue}{\mathrm{v}}_j)$ implies that these two vertices are neighbors to each other, and
	both agents, correspondingly, are able to utilize each other’s information. Denote ${\cal N}_i$ as the neighbor set of agent $i$, which is considered not to belong to this set.
	\subsection{PPC Method}
	\par {Applying PPC method} can ensure that the system exhibits prescribed dynamic and steady-state performance indicators, which are encapsulated in a performance function $\rho(t)$. That is, the scalar error $e(t)$ of the system will converge to any given residual set $(-\rho_\infty,\rho_\infty)$ at a rate not less than a predetermined rate $\lambda$. Next, \textcolor{blue}{it is necessary to} review this idea, which mainly draws on the discussions of relevant authors in \cite{Bechlioulis2008, Bechlioulis2009}.
	\par The design of a prescribed performance controller depends on the selection of smooth and bounded performance functions, {such as,}
	\begin{equation}
		\label{5}
		\rho(t)=(\rho_0-\rho_\infty)e^{-\lambda t}+\rho_\infty,
	\end{equation}
	where $\rho_0$, $\rho_\infty$ and $\lambda$ are positive constants, $\rho_0$, furthermore, is
	strictly selected so that $\rho_0> |e(0)|$. In this way, the scalar error $e(t)$ is strictly limited to evolving within the predetermined region, i.e,
	$$-\rho(t)<e(t)<\rho(t),\forall t\ge 0.$$
	And this also indicates that $e(t)$ will not decay at a rate lower than $\lambda$ and ultimately ensure steady-state performance.
	\par Next, let’s review the error transformation technology for future use. Define a modulated error $\xi(t)=\frac{e(t)}{\rho(t)}$. Applying a strictly increasing, odd and bijective mapping $T$:$(-1,1)\to(-\infty,+\infty)$ with $T(0)=0$, taking the following as an example, \textcolor{blue}{one can obtain that}
	\begin{equation}
		\label{6}
		T(\xi)=\frac{1}{2}\mathrm{ln}(\frac{1+\xi}{1-\xi}),
	\end{equation}
	whose Jacobian derivative is found as follow,
	\begin{equation}
		\label{7}
		J_T(\xi)=\frac{\mathrm{d}T(\xi)}{\mathrm{d}\xi}=\frac{1}{1-{\xi}^2},
	\end{equation}
	such that  it always behaves positively. In addition, (\ref{6}) and (\ref{7}) have the following properties.
	\begin{lemma}
		\label{lepre}
		\cite{Bechlioulis2010} For any $\mathrm{K}$, $\mathrm{F}>0$, there must be
		$$-\mathrm{K}T(\xi)J_T(\xi)\xi+\mathrm{F}<0$$
		for all $\vert T(\xi)\vert>\bar\epsilon$ with a certain positive constant $\bar\epsilon$.
	\end{lemma}
	\subsection{Dynamical Systems}
	Consider the following initial value problem for a dynamical system,
	\begin{equation}
		\label{8}
		\dot x=f(t,x),\;x(0)=x^0\in\Omega_x,
	\end{equation}
	where $f:\mathbb{R}_{+} \times \Omega_x \rightarrow \mathbb{R}^n$ and $\Omega_x\subset \mathbb{R}^n$ is an open and nonempty set.
	\begin{definition}
		\cite{Sontag1998} For the initial-value problem (\ref{8}), there is a maximal solution $x(t)$ if it cannot be extended to the right.
	\end{definition}
	\begin{lemma}
		\label{lem1}
		\cite{Sontag1998} Consider the initial-value problem (\ref{8}) for $x\in \Omega_x$. If $f(t,x)$ exhibits local Lipschitz on $x$ for $t\le 0$, piecewise {continuous} and locally integrable on $t$ at each $x\in \Omega_x$, the dynamical system \eqref{8} must have a maximal solution $x(t)\in\Omega_x$ on a time interval $[0,t_\mathrm{max})$ with $t_\mathrm{max}>0$.
	\end{lemma}
	\begin{proposition}
		\label{pro2}
		\cite{Sontag1998} Based on Lemma \ref{lem1}, for a finite time interval $[0,t_\mathrm{max})$ with $t_\mathrm{max}<\infty$ {and any compact subset $\Omega^{'}_x\in \Omega_x$}, there must exist a time instant $t'\in [0,t_\mathrm{max})$ such that $x(t')\in\Omega_x^{'}$.
	\end{proposition}
	\section{Design of SoC Balancing Scheme Under Intermittent Dynamics}
	{In this section, two} estimation schemes with prescribed performance driven by intermittent dynamics {are} designed to estimate the {power} sharing state and SoC balance state. Subsequently, {these} schemes {are} further developed to achieve intermittent communication. Before this, it is necessary to explain the load to be supported.
	\begin{assumption}
		\label{assum2}
		In this paper, the total load is continuous, bounded, and has a first derivative.
	\end{assumption}
	\par Referring to Assumption \ref{assum2}, the load that each node needs to bear is continuous, bounded, and has a first derivative.
	\par Considering that continuous communication among agents cannot be achieved through the underlying network, an intermittent communication scheme with an {self} triggering mechanism is designed in this section to accurately estimate the average state of {batteries. Prior to this, some preliminary work, namely an event triggering mechanism, is designed to determine the frequency of updating consensus item}.
	\subsection{Design of {Estimators} under Intermittent {Dynamics}}
	\par The {two estimators} with event triggering mechanisms adopted are as follows 
	{\begin{subequations}
	   \begin{equation}
		   \label{9a}
		   \dot{z}_i^P=-\mathrm{k}^P_{\mathrm{r}}(\hat{P}_{\mathrm{a},i}-D_i)-\mathrm{k}^P\sum_{j\in\mathcal{N}_i}h_{ij}^P(t_{P,k}^{ij})\textcolor{blue}{,}
	\end{equation}
	\begin{equation}
		\label{9b}
		\hat{P}_{\mathrm{a},i}=z_i^P+D_i\textcolor{blue}{,}
	\end{equation}
	\begin{equation}
		\label{9c}
		\dot{z}_i^E=-\mathrm{k}^E_\mathrm{r}(\hat{E}_{\mathrm{a},i}-E_i)-\mathrm{k}^E\sum_{j\in\mathcal{N}_i}h_{ij}^E(t_{E,k}^{ij})\textcolor{blue}{,}
	\end{equation}
	\begin{equation}
		\label{9d}
		\hat{E}_{\text{a},i}=z_i^E+E_i\textcolor{blue}{,}
	\end{equation}
    \end{subequations}}
 where {$h_{ij}^P(t)=\frac{1}{\rho^P_{ij}}J_{T}(\frac{\hat{P}_{\mathrm{a},i}-\hat{P}_{\mathrm{a},j}}{\textcolor{blue}{\rho^P_{ij}}})T(\frac{\hat{P}_{\mathrm{a},i}-\hat{P}_{\mathrm{a},j}}{\textcolor{blue}{\rho^P_{ij}}}),$} {$h_{ij}^E(t)=\frac{1}{\rho^E_{ij}}J_{T}(\frac{\hat{E}_{\mathrm{a},i}-\hat{E}_{\mathrm{a},j}}{\textcolor{blue}{\rho^E_{ij}}})T(\frac{\hat{E}_{\mathrm{a},i}-\hat{E}_{\mathrm{a},j}}{\textcolor{blue}{\rho^E_{ij}}}),$ $\rho^P_{ij}$ and $\rho^E_{ij}$} are performance functions defined by \eqref{5}. {Here, event triggering mechanisms as follow are provided for the designed two estimators}
\begin{subequations}
	\label{10}
	    \begin{equation}
		\label{10a}
		{t_{P,k+1}^{ij}=\inf\limits_{t>t_{P,k}^{ij}}}\left\{t\in\mathbb{R}_+:{\tilde{h}_{P,ij}^2(t)}\geq\mathrm{c}_0+\mathrm{c}_1\mathrm{e}^{-\alpha t}\right\},
	\end{equation}
	    \begin{equation}
		\label{10b}
		{t_{E,k+1}^{ij}=\inf\limits_{t>t_{E,k}^{ij}}\left\{t\in\mathbb{R}_+:{\tilde{h}_{E,ij}^2(t)}\geq\mathrm{d}_0+\mathrm{d}_1\mathrm{e}^{-\beta t}\right\},}
	\end{equation}
\end{subequations}
    where the measurement \textcolor{blue}{errors} ${\tilde{h}_{P,ij}(t)}$ and ${\tilde{h}_{E,ij}(t)}$ are defined as
    $${\tilde{h}_{P,ij}(t)}={h_{ij}^P}(t_{P,k}^{ij})-{h_{ij}^P}(t),\forall t\in\left[t_{P,k}^{ij},t_{P,k+1}^{ij}\right),$$
    $${\tilde{h}_{E,ij}(t)={h_{ij}^E}(t_{E,k}^{ij})-{h_{ij}^E}(t),\forall t\in\left[t_{E,k}^{ij},t_{E,k+1}^{ij}\right),}$$
    and $\mathrm{c}_0$, $\mathrm{c}_1$, $\alpha$, {$\mathrm{d}_0$, $\mathrm{d}_1$, and $\beta$} are positive parameters to be designed. 
    {\begin{remark}   		
    From \eqref {10}, it can be seen that the event triggering scheme adopted is an asynchronous triggering mechanism based on a directed communication link. More precisely, whether agent $i$ sends its own state to another agent $j$ depends on the state difference between agent $i$ and $j$, which is different from the requirement in \cite{9770469} that both parties exchange information at the trigger time. In addition, for a BESS, the SoC balance state estimator and power sharing state estimator are also asynchronously triggered. That is to say, for the agent $i$ managing BESS $i$, ${\hat E}_{a,i}$ and ${\hat P}_{a,i}$ do not need to be transmitted simultaneously, which is also different from \cite{9770469}.
    \end{remark}}
    \par {Since the two estimators \textcolor{blue}{are in} the same form, the properties of the power sharing estimator \textcolor{blue}{is only analyzed as follow}, and then a theorem about the SoC balance state estimator \textcolor{blue}{is attached without any proof}. A theorem and its proof regarding the designed estimator are presented as follows.}
	\begin{theorem}
		\label{TH1}
		Consider a microgrid containing $\mathrm{n}$ BESSs. Under the control strategy {\eqref{9a} and \eqref{9b} with \eqref{10}, the average load value} can be estimated for all $i\in\{1,\cdots,\mathrm{n}\}$ with steady-state error as follow,
		\begin{equation}
			\label{11}
			\lim\limits_{t\to+\infty}\vert \hat P_{a,i}-D_a\vert\le \frac{\mathrm{Diam}({\cal G})\rho_\infty}{2},
		\end{equation}
		and Zeno behavior can be avoided, where $\mathrm{Diam}({\cal G})$ denotes the diameter of ${\cal G}$. In particular, {consensus error also has prescribed dynamic and steady performance defined in the performance function $\rho^P_{ij}(t)$}.
	\end{theorem}
	\begin{proof}
		Let’s first define some stack vectors used in the proof. Assuming there are $\mathrm{m}$ edges in ${\cal G}$ and edge $(i,j)$ is associated with the $l^{th}$ edge, define $\delta_{l}^{P}=\hat{P}_{\mathrm{a},i}-\hat{P}_{\mathrm{a},j}$ and its stack vector $\delta^{P}=\mathrm{col}(\delta_{1}^{P},\delta_{2}^{P},\cdots,\delta_{\mathrm{m}}^{P})$. Hence, the
		performance function {$\rho^P_{ij}$} associated with $\delta_l$ can be denoted as $\rho_l$. Let $\mathrm{B}$ denote the incidence matrix of vertices and edges in $\cal G$. Thus, $\delta^P=\mathrm{B}\hat{P}_{\mathrm{a},i}$. And define a normalized error vector $\xi^P$ as
		$$\xi^P=\mathrm{col}(\frac{\delta_1^P}{\rho_1(t)},\cdots,\frac{\delta_\mathrm{m}^P(t)}{{\rho_\mathrm{m}(t)}})=R^{-1}(t)\delta^P,$$
		where $R=\mathrm{diag}([\rho_1,\cdots,\rho_\mathrm{m}])$. Applying the mapping \eqref{6}, one can obtain
		$$T^P=\mathrm{col}(T(\xi_1^P),T(\xi_2^P),\cdots,T(\xi_\mathrm{m}^P)).$$
		Then the Jacobian derivatives of $T^P$ can be derived from \eqref{7} as follow,
		$$J_T^P=\mathrm{diag}([J_T(\xi_1^P),\cdots,J_T(\xi_\mathrm{m}^P)]).$$
		Under the action of \eqref{9a}, the power sharing state {estimator} follows the following dynamic in compact form
		\begin{equation}
			\label{12}
			\dot{\hat{P}}_{\mathrm{a}}=\dot{D}-{\mathrm{k}^P_{\mathrm{r}}}(\hat{P}_{\mathrm{a}}-D)-{\mathrm{k}^P}\mathrm{B}R^{-1}J_{T}^{P}T^{P}-\mathrm{kB}\tilde{h}(t),
		\end{equation}
		where $\tilde{h}(t)=\mathrm{col}(\sum_{j\in\mathcal{N}_{i}}\tilde{h}_{1j}(t),\cdots,\sum_{j\in\mathcal{N}_{i}}\tilde{h}_{\mathrm{n}j}(t))$, $\hat{P}_{\mathrm{a}}=\mathrm{col}(\hat{P}_{\mathrm{a},1},\hat{P}_{\mathrm{a},2},\cdots,\hat{P}_{\mathrm{a},\mathrm{n}})$. Based on this, find the first derivative of $\xi^{P}$ as
		\begin{equation}
			\label{13}
			\begin{aligned}
				\dot{\xi}^P&=f(t,\xi^P) \\
				&=R^{-1}(t)(\dot\delta^{P}-\dot{R}(t)\xi^{P}) \\
				&=R^{-1}(t)(\mathrm{B}^{\mathrm{T}}(\dot{D}-{\mathrm{k}^P_{\mathrm{r}}}(\hat{P}_{\mathrm{a}}-D)-{\mathrm{k}^P}\mathrm{B}R^{-1}J_{T}^{P}T^{P} \\
				&-{\mathrm{k}^P}\mathrm{B}\tilde{h}(t))-\dot{R}(t)\xi^{P}).
			\end{aligned}
		\end{equation}
		In addition, define an open set as follow
		$$\Omega_\xi=\underbrace{(-1,1)\times(-1,1)\cdots\times(-1,1)}_{\text{m-times}}.$$
		\par Next, Theorem \ref{Th1} will be proved in three phases. Firstly, \textcolor{blue}{it can be proved} that the unique maximal solution of the dynamical system \eqref{13} exists over the set $\Omega_\xi$ in a finite time set $[0, t_\mathrm{max})$. Secondly, we provide proof that the solution $\xi^P$of the dynamical system \eqref{13} always evolves over a set $\Omega_\xi$ in a finite time set $[0, t_\mathrm{max})$, and at the same time, apply the proof by contradiction to deduce that the above conclusion still holds when $t_\mathrm{max} = +\infty$. In this way, one can get $|\delta_l^P|<\rho_l$ for $\forall t > 0 $ and $l \in\{1, 2, \cdots, \mathrm{m}\}$. At this point, the consensus of power sharing estimation can be achieved with the prescribed dynamic and steady-state performance specifications. Then, the asymptotic convergence characteristics of the tracking error $e^{\dot{P}}=\frac{1^{T}}{\mathrm{n}}(\hat{P}_{\mathrm{a}}-D)$ of the designed estimator are demonstrated. Finally, Zeno behavior of the designed estimator is analyzed by using the the Proof by contradiction. Thus, Theorem \ref{TH1} is fully proven.
		\par \emph{Phase A}. By pre-setting the performance function $\rho_l(t)$ to satisfy $\rho_l(0) > |\delta_l(0)|$ for all $l\in\{1, 2, \cdots, \mathrm{m}\},$ it can be obtained that $|\xi_l^P(0)|=\frac{\delta_l(0)}{\rho_l(0)}<1$, which implies $\textcolor{blue}{\xi^P(0)}\in\Omega_\xi $. According to \eqref{13}, $f(t, \xi^P)$ behaves piecewise continuous, locally integrable on $t$, and locally Lipschitz on $\xi^P$. Hence, referring to Lemma \ref{lem1}, \eqref{13} must have a maximal solution on a finite time interval $[0, t_\mathrm{max})$, and $\xi^P\in\Omega_\xi$ for $t\in[0, t_\mathrm{max})$.
		\par \emph{Phase B}. {According to the result in \emph{Phase A}}, $-1<\xi^P_l<1$ for $\forall t\in[0, t_\mathrm{max})$ and $\forall l\in \{1, 2,\cdots, \mathrm{m}\}$. Next, {select} a Lyapunov function as shown below
		$$V=\frac{1}{2}(T^P)^\mathrm{T}T^P+\frac{\mathrm{kmc}_1}{2\alpha}\|\mathrm{B}\|^2\mathrm{e}^{-\alpha t},$$
		and find its derivative as
		$$\begin{aligned}
			\dot{V}=& (T^{P})^{\mathrm{T}}J_{T}^{P}R^{-1}(t)(\mathrm{B}^{\mathrm{T}}(\dot{D}-{\mathrm{k}^P_{\mathrm{r}}}(\hat{P}_{\mathrm{a}}-D)  \\
			&-{\mathrm{k}^P}\mathrm{B}R^{-1}(t)J_{T}^{P}T^{P}-{\mathrm{k}^P}\mathrm{B}\tilde{h}(t))-\dot{R}(t)\xi^{P}) \\
			&-\frac{{\mathrm{k}^P}\mathrm{mc}_1}2\|\mathrm{B}\|^2\mathrm{e}^{-\alpha t} \\
			&\text{=} -{\mathrm{k}^P_{\mathrm{r}}}(T^{P})^{\mathrm{T}}J_{T}^{P}R^{-1}(t)\mathrm{B}^{\mathrm{T}}\hat{P}_{\mathrm{a}}-(T^{P})^{\mathrm{T}}J_{T}^{P}R^{-1}(t)\dot{R}(t)\xi^{P}  \\
			&-{\mathrm{k}^P}(T^P)^\mathrm{T}J_T^PR^{-1}(t)\mathrm{B}^\mathrm{T}\mathrm{B}R^{-1}(t)J_T^PT^P \\
			&-{\mathrm{k}^P}(T^{P})^{\mathrm{T}}J_{T}^{P}R^{-1}(t)\mathrm{B}^{\mathrm{T}}\mathrm{B}\tilde{h}(t)-\frac{{\mathrm{k}^P}\mathrm{mc}_{1}}{2}\|\mathrm{B}\|^{2}\mathrm{e}^{-\alpha t} \\
			&+(T^P)^\mathrm{T}J_TR^{-1}(t)\mathrm{B}^\mathrm{T}({\mathrm{k}^P_{\mathrm{r}}}D+\dot{D}).
		\end{aligned}$$
		Applying $\xi^{P}=R^{-1}(t)\mathrm{B}^{\mathrm{T}}\hat{P}_{\mathrm{a}},|\frac{\dot{\rho}_{l}}{\rho_{l}}|<\lambda $, and properties of total load $D$, \textcolor{blue}{one} can get
		$$\begin{aligned}
			\dot{V}\leq & -{\mathrm{k}^P_{\mathrm{r}}}(T^{P})^{\mathrm{T}}J_{T}^{P}\xi^{P}+\lambda(T^{P})^{\mathrm{T}}J_{T}^{P}\xi^{P}  \\
			&-{\mathrm{k}^P}(T^{P})^{T}J_{T}^{P}R^{-1}(t)\mathrm{B}^{\mathrm{T}}\mathrm{B}R^{-1}(t)J_{T}^{P}T^{P} \\
			&+\sup_{t\geq0}(\|{\mathrm{k}^P_{\mathrm{r}}}D+\dot{D}\|)\|T^{P}J_{T}^{P}R^{-1}(t)\mathrm{B}^{\mathrm{T}}\| \\
			&+{\mathrm{k}^P}\|\mathrm{B}\|\|\tilde{h}(t)\|\|T^{P}J_{T}^{P}R^{-1}(t)\mathrm{B}^{\mathrm{T}}\| \\
			&-\frac{{\mathrm{k}^P}\mathrm{mc}_1}2\|\mathrm{B}\|^2\mathrm{e}^{-\alpha t}.
		\end{aligned}$$
		Considering the event triggering mechanism \eqref{10}, it can be obtained that
		\begin{equation}
			\label{14}
			\|\tilde{h}(t)\|^2\leq\mathrm{mc}_0+\mathrm{mc}_1\mathrm{e}^{-\alpha t},\forall t\in\underset{j\in\mathcal{N}_i}{\cup}[t_k^{ij},t_{k+1}^{ij}).
		\end{equation}
		Thus, according to the Young’s inequality, the derivative of $V$ can be derived as
		\begin{equation}
			\label{15}
			\begin{aligned}
				\dot{V}\leq&-({\mathrm{k}^P_{\mathrm{r}}}-\lambda)(T^{P})^{\mathrm{T}}J_{T}^{P}\xi^{P}\\&+\frac{\sup_{t\geq0}(\|{\mathrm{k}^P_{\mathrm{r}}}D+\dot{D}\|)^{2}}{2\mathrm{k}}+\frac{{\mathrm{k}^P}\mathrm{mc}_{0}}{2}\|\mathrm{B}\|^{2}.
			\end{aligned}
		\end{equation}
		If $\mathrm{k_{r}}>\lambda $, according to Lemma \ref{lem1}, there must be a positive
		constant $\bar{\epsilon}$ such that $\dot{V}\leq 0$ for $\|T^P\|>\bar{\epsilon}$, which dictates that
		$|T_l(\xi^P)|\leq\epsilon^*=\max\{\|\xi^P(0)\|,\bar{\epsilon}\},\forall l\in\{1,\cdots,\text{m}\},$
		within $t\in[0, t_\mathrm{max})$. Applying the inverse function of the hyperbolic tangent function, we have
		$$|\xi_l^P(t)|\leq\xi^*=\tanh(\epsilon^*)<1,$$
		which implies that $\xi_l^P$ always evolves in $\Omega_\xi$ for $\forall t \in [0, t_\mathrm{max})$ and $l\in\{1, 2, \cdots, \mathrm{m}\}$.
		\par Thus, there must exist a compact and nonempty subset of $\Omega_\xi$, denoted as $\Omega'_\xi$, such that $\xi_l^P\in\Omega'_\xi$ for $\forall t \in [0, t_\mathrm{max})$. However, according to Proposition \ref{pro2}, there must exist \textcolor{blue}{an} instant $t'$ 
		such that $\xi^{P}(t^{\prime})\notin\Omega_{\xi}^{'}$ if $t_{\mathrm{max}}<+\infty $. At this point, a
		clear contradiction arises. As a result, $t_\mathrm{max}$ can be extended to $+\infty$. That is, $|\delta_l^P|<\rho_l$ always holds for $\forall t\in[0,+\infty)$ and $l\in\{1,2,\cdots,\mathrm{m}\}$. Up to this point, the consensus errors are successfully constrained within the specified range.
		\par \emph{Phase C}. From \emph{Phase B}, it can be inferred that the average consensus error can be expressed as
		$$\lim\limits_{t\to+\infty}|\hat{P}_{\mathrm{a},i}-\frac{\mathbf{1}^\mathrm{T}}{\mathrm{n}}\hat{P}_{\mathrm{a}}|\leq\frac{\mathrm{Diam}(\mathcal{G})\rho_{\infty}}{2},i=1,\cdots,\mathrm{n}$$
		with exponential convergence rate $\lambda$. The tracking errors behave $\dot{e}^{P}=-{\mathrm{k}^P_{\mathrm{r}}}e^{P}$. Hence, only if ${\mathrm{k}^P_{\mathrm{r}}>\lambda}$, the steady-state
		performance complies with the following description
		$$\lim\limits_{t\to+\infty}|\hat{P}_{\mathrm{a},i}-\frac{\mathbf{1}^\mathrm{T}}{\mathrm{n}}D|\leq\frac{\mathrm{Diam}(\mathcal{G})\rho_\infty}{2},i=1,\cdots,\mathrm{n}$$
		with exponential convergence rate $\lambda$.
		\par \emph{Phase D}. To ensure the feasibility of the algorithm in practical applications, Zeno behavior needs to be excluded. So, next, we will use the Proof by contradiction to analyze the number of triggers in any finite time {interval}.
		\par Assume that there exists a finite time instant $t_{e}$ such that $\lim_{k\to+\infty}t_{k}^{ij}=t_{e}$ for any pair $(i, j)$ agents. When Zeno behavior occurs,
		$$|{\tilde{h}_{P,ij}}(t_{k+1}^{ij-})|=\sqrt{\mathrm c_0+\mathrm c_1\mathrm e^{-\alpha t_{k+1}^{ij}}}\geq\sqrt{\mathrm c_0+\mathrm c_1\mathrm e^{-\alpha t_e}}.$$
		Besides, $\tilde{h}_{ij}(t)$ for any pair $(i, j)$ behaves continuously differentiable, and $\dot{\tilde{h}}_{ij}(t)$ is bounded, i.e., $\|{\dot{\tilde h}_{P,ij}}(t)\|<\bar{h}_{ij}$ for $\forall t>0$. In view of this, we arrive at
		$$|{{\tilde h}_{P,ij}}(t_{k+1}^{ij-})|\leq\int_{t_k^{ij}}^{t_{k+1}^{ij}}|{\dot{\tilde h}_{P,ij}}(\tau)|d\tau\leq(t_{k+1}^{ij}-t_k^{ij})\bar{h}_{ij}.$$
		Combining the above two inequalities, we can get
		\begin{equation}
			\label{16}
			t_{k+1}^{ij}-t_k^{ij}\geq\frac{\sqrt{\mathrm{c}_0+\mathrm{c}_1\mathrm{e}^{-\alpha t_e}}}{\bar{h}_{ij}},
		\end{equation}
		which implies that there will not be an infinite number of triggers occurring within any finite time period. Hence, the proposed scheme behaves Zeno-free.
		\par With the help of \eqref{14}, the interval between any two adjacent triggering instants satisfies the following inequality
		\begin{equation}
			\label{17}
			t_{k+1}^{ij}-t_k^{ij}\geq\frac{\sqrt{\mathrm{c}_0+\mathrm{c}_1\mathrm{e}^{-\alpha t_k^{ij}}}}{\bar{h}_{ij}},
		\end{equation}
		which means that the triggering scheme \eqref{10} will not cause Zeno behavior.
		\par So far, the proof of Theorem \ref{TH1} \textcolor{blue}{is completed}. \hfill$\blacksquare$
	\end{proof}
	{\begin{remark}
			According to Theorem \ref{TH1}, it can be seen that selecting different values of $\rho_0$, $\rho_\infty$, and $\lambda$ in the performance function $\rho^P(t)$ can to some extent adjust the dynamic and steady-state performance of the involved schemes. Next, let's first discuss the impact of $\mathrm{k}_\mathrm{r}^P$ and $\mathrm{k}^P$ on ${\hat P}_a$. According to \eqref{15}, an increase in $\mathrm{k}_\mathrm{r}^P$ can lead to a smaller $\dot V$, which is beneficial for achieving faster consensus on the estimated power sharing state. Furthermore, according to $\dot{e}^{P}=-\mathrm{k}_\mathrm{r}^Pe^{P}$, increasing $\mathrm{k}_\mathrm{r}^P$ can accelerate the convergence of the estimated average power sharing state to the average load. However, the impact of $\mathrm{k}^P$ on the results is not clear. According to \eqref{15}, the impact of increasing $\mathrm{k}^P$ on $\dot V$ should be further discussed in conjunction with the configuration of the load and communication network. Generally speaking, for small and medium-sized power networks, $\frac{\sup_{t\geq0}(\|\mathrm{k}^P_{\mathrm{r}}D+\dot{D}\|)^{2}}{2\mathrm{k}^P}\gg\frac{\mathrm{k}^P\mathrm{mc}_{0}}{2}\|\mathrm{B}\|^{2}$. So, an arbitrary conclusion can be drawn that increasing $\mathrm{k}^P$ can lead to a smaller $\dot V$, thereby accelerating the consensus rate of ${\hat P}_{\mathrm{a},i}$ for $\forall i\in\{1,2,\cdots,\mathrm{n}\}$.
	\end{remark}}
	\begin{remark}
		Although Zeno-free behavior can be guaranteed, {the parameters of the event triggering mechanism \eqref{10} have a significant impact on the triggering frequency and dynamic response of $\hat{P}_{\text{a},i}$}. Below is an explanation {for} these. Given by a smaller parameter $\mathrm{c}_0$, the measured error $\tilde{h}_{ij}(t)$ and the upper bound $\bar{\epsilon}$ of the error \textcolor{blue}{$T^P_i$ for $\forall i\in\{1,2,\cdots,\mathrm{m}\}$} become smaller observed from \eqref{10} and \eqref{14}\textcolor{blue}{, which combing} with Lemma \ref{lepre} \textcolor{blue}{results} in a smaller the error $T^P_i$. But the impact of an increasing $\mathrm{c}_0$ on the results is uncertain, which can be inferred from \eqref{15}. Specifically, the increase in $\mathrm{c}_0$ leads to an increase in $\tilde{h}_{ij}(t)$ and $\bar{h}_{ij}$, resulting in both the numerator and denominator of the lower bound of the trigger time interval in \eqref{17} increasing simultaneously. Further, due to the fact that the actual response is constrained by the load of each bus and its derivative, the trial and error method \textcolor{blue}{needs to} be applied to illustrate the impact on the triggering frequency.
	\end{remark}
	\par {Similarly, there is a theorem as follow for the designed SoC balance state estimator \eqref{9c} and \eqref{9d} with \eqref{10}. For simplicity purposes, \textcolor{blue}{the} proof is omitted due to its similarity to Theorem \ref{TH1}.}
	\begin{theorem}
		\label{TH21}
		Consider a microgrid containing $\mathrm{n}$ BESSs, and SoC and output power of each BESS are subjected to \eqref{1} and \eqref{2} respectively. Under the control strategy {\eqref{9c} and \eqref{9d} with \eqref{10}, the average SoC state} can be estimated {by agent $i$} for all $i\in\{1,\cdots,n\}$ with steady-state error as follow,
		\begin{equation}
			\label{18}
			\lim\limits_{t\to+\infty}\vert \hat E_{a,i}-E_a\vert\le \frac{\mathrm{Diam}(G)\rho_\infty}{2}
		\end{equation}
		and Zeno behavior can be avoided. In particular, {consensus error also has prescribed dynamic and steady performance defined in the performance function $\rho^E_{ij}(t)$}.
	\end{theorem}
	\subsection{A Self Triggering Mechanism for Intermittent Communication}
	\par Although the event triggering mechanism \eqref{10} can ensure the prescribed performance of the estimators without Zeno behavior, it requires continuous real-time monitoring of measurement errors, which is not {a} truly intermittent communication {scheme}. So, another conservative strategy \textcolor{blue}{is developed here}, i.e.,
	\begin{equation}
		\label{19}
		{t_{P,k+1}^{ij}}=\inf_{t>{t_{P,k}^{ij}}}\left\{t\in\mathbb{R}_{+}:(t-{t_{P,k}^{ij}})^{2}\overline{h}_{ij}^{2}\geq\mathrm{c}_{0}+\mathrm{c}_{1}\mathrm{e}^{-\alpha t}\right\},
	\end{equation}
	{where $\overline{h}_{ij}$ is the upper bound of $\|\dot{\tilde{h}}_{ij}(t)\|$ according to \eqref{16}.} Additionally, \eqref{19}, which indicates that the next triggering instant ${t_{P,k+1}^{ij}}$ only depends on the information at the current instant ${t_{P,k}^{ij}}$, is a self triggering mechanism that does not require continuous communication {and can also be called a time-triggered mechanism. Based on Theorem \ref{TH1} and \ref{TH21}, and \eqref{17}, the following corollary is given without proof.}
	{\begin{corollary}
			Under the control protocol \eqref{9a}-\eqref{9d} with \eqref{19}, the power sharing state and SoC balance state can be well estimated, i.e.,
			$$\lim\limits_{t\to+\infty}\vert \hat E_{a,i}-E_a\vert\le \frac{\mathrm{Diam}(G)\rho_\infty}{2},$$
			$$\lim\limits_{t\to+\infty}\vert \hat P_{a,i}-P_a\vert\le \frac{\mathrm{Diam}(G)\rho_\infty}{2}.$$
			Meanwhile, consensus error has the prescribed dynamic and steady performance, and Zeno behavior can be excluded.
	\end{corollary}}
	\begin{remark}
		\label{rem2}
		Although $\hat{P}_{\mathrm{a},i}$ and $\hat{E}_{\mathrm{a},i}$ exhibit Zeno-free behavior under the self-triggering mechanism \eqref{19}, communication may be triggered more and more frequently over time. This is because the time interval that satisfies \eqref{19} becomes shorter over time, which leads to earlier communication times. So, this scheme demonstrates conservatism. Additionally, based on simple monotonicity judgment, the measurement error is positively correlated with the modulation error. This indicates that $\overline{h}_{ij}$ should satisfy $\overline{h}_{ij}\geq\tilde{h}(t)$ with $\xi^{P}=\bar{\xi}$. Following this approach, $\frac{\sup_{t\geq0}(\|{\mathrm{k}^P_\mathrm{r}}\dot{D}+\dot{D}\|)^{2}}{2\mathrm{k}}+\frac{{\mathrm{k}^P}\mathrm{mc_{0}}}{2}\|\mathrm{B}\|^{2}$ is needed to calculate $\bar{\xi}$ through \eqref{15}. Otherwise, this self triggering mechanism may become more conservative.
	\end{remark}
	\section{Some Simulation Cases}
	\par In this section, {five} cases are designed to test the performance of the designed schemes. In Case 1, the event triggering scheme is tested to demonstrate its effectiveness. To illustrate the progressiveness, {an existing event triggering scheme} developed in \cite{Qian2023} is used to compare with the event triggering scheme {in Case 2}. {Next,} the {Play-and-Plug function and scalability are scheduled for testing in Case 3 and 4, respectively}. {Finally, the involved self triggering scheme is tested in Case 5.} The resistance network used and the communication network of agents are shown in Fig. \ref{comm}. {This DC power system is controlled by droop controllers and adopts a single bus with four sections, each containing a BESS and load.}
	\begin{figure}
		\centering
		\includegraphics[width=8cm]{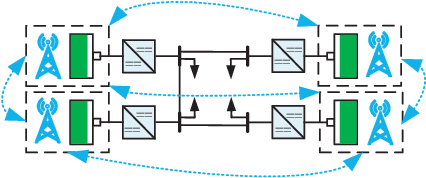}
		\caption{Four BESSs in a resistance network and their communication graph
			topology}
		\label{comm}
	\end{figure}
	\begin{figure}
		\centering
		\includegraphics[width=8cm]{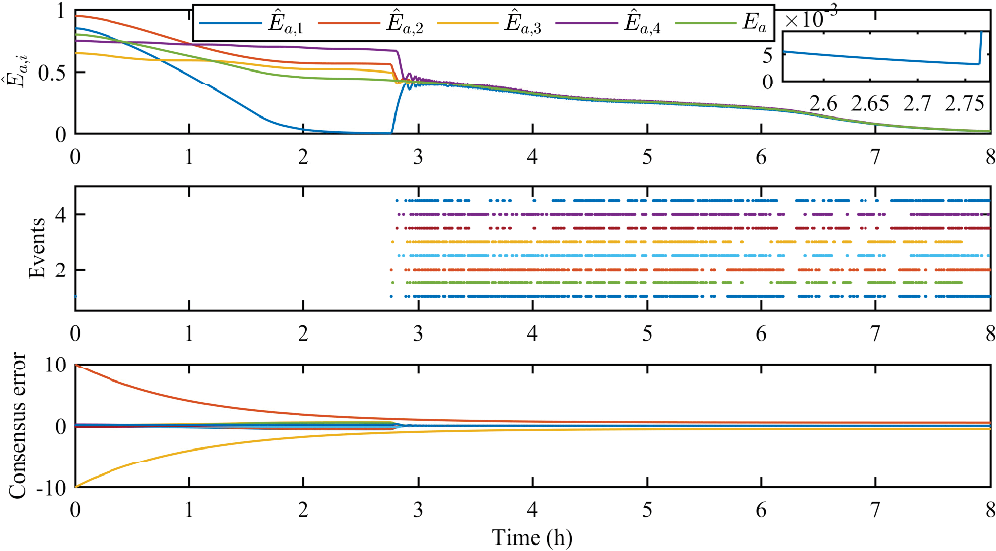}
		\caption{Evolution of SoC balance estimation, events, and consensus error in Case 1.}
		\label{case1a}
	\end{figure}
	\begin{figure}
		\centering
		\includegraphics[width=8cm]{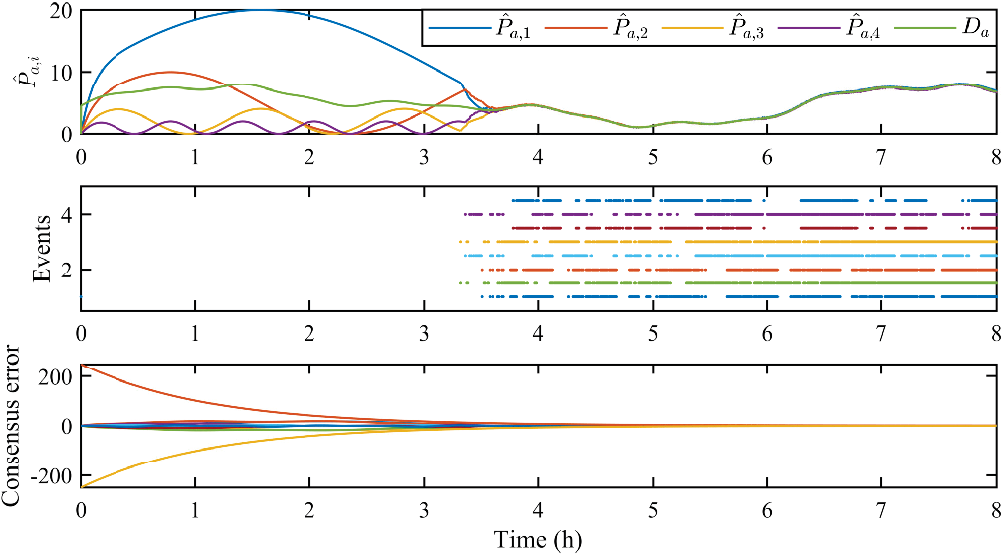}
		\caption{Evolution of power sharing estimation, events, and consensus error in Case 1.}
		\label{case1b}
	\end{figure}
	\begin{figure}
		\centering
		\includegraphics[width=8cm]{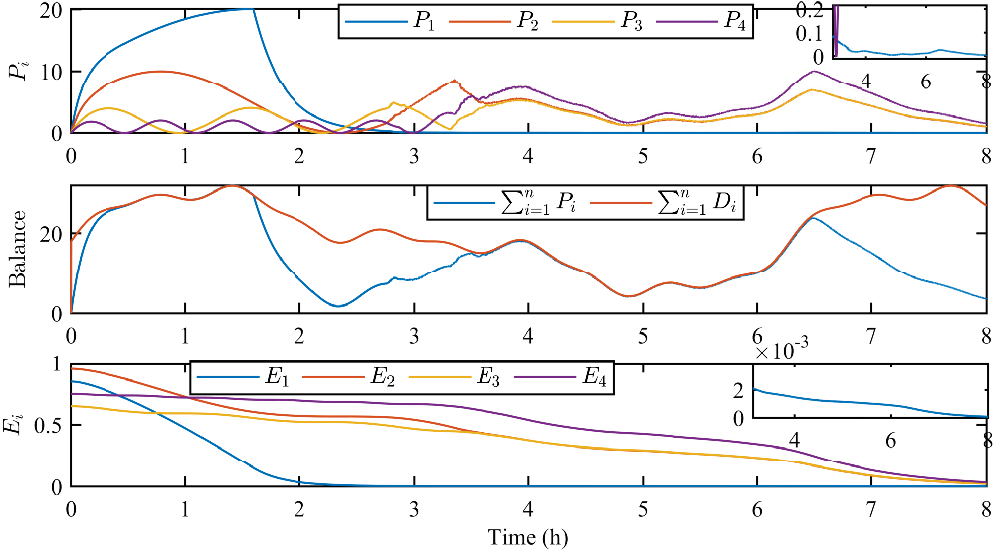}
		\caption{Evolution of output power, supply-demand balance, and SoC in Case 1.}
		\label{case1c}
	\end{figure}
	\subsection{Case 1. The {Test on the Designed Scheme with the} Event-triggering Mechanism {\eqref{10}}}
	\par Under the event triggering mechanism \eqref{10} \textcolor{blue}{and} the designed schemes \eqref{9a}-\eqref{9d}, the simulation results are shown in \textcolor{blue}{Figs.} \ref{case1a}-\ref{case1c}. {Among of them,} Fig. \ref{case1a} demonstrates events per edge, consensus error, and the evolution of {power} sharing state estimate. SoC balance state estimate, consensus error, and its events per edge are shown in Fig. \ref{case1b}. The evolution of the SoC and output power of each BESS, as well as the supply-demand balance of the entire network, is shown in Fig. \ref{case1c}.
	\par Overall, from \textcolor{blue}{Figs.} \ref{case1a} and \ref{case1b}, the {power} sharing and SoC balance states can be well estimated by each agent. It is worth mentioning that {for BESS $i$}, $\hat{P}_{\mathrm{a},i}$ and $\hat{E}_{\mathrm{a},i}$ {are} not synchronously transmitted in the communication network, {which} is drawn from the events per edge of two estimated variables. However, both $\hat{P}_{\mathrm{a},i}$ and $\hat{E}_{\mathrm{a},i}$ behave oscillations. This bad behavior of $\hat{E}_{\mathrm{a},i}$ is particularly severe. {Besides,} whether it is the estimated power average state or the estimated SoC average state, the consensus error can {always evolve within the predefined region and converge at a rate not less than the predefined rate}.
	\par {From Fig. \ref{case1c}, it can be easily observed that 
		the supply-demand balance can be well maintained between approximately $t=0.5h$ and $t=1.5h$. This is because there is no communication between agents, and ${\hat E}_{a,i}$ gradually converges to $E_i$, while ${\hat P}_{a,i}$ gradually converges to $P_i$. According to \eqref{Pcal}, at this point, $P_i$ gradually converges to $D_i$, thereby ensuring supply-demand balance. From $t=1.5h$ to $t=3h$, BESS 1 is about to run out according to $E_1$, causing an imbalance between supply and demand. Subsequently, as communication resumed, the SoC balance state and power sharing state are well estimated, thereby ensuring supply-demand balance. At approximately $t=6.5h$, the supply-demand imbalance occurs again, due to the low energy level of BESSs in this network, resulting in insufficient power supply. Although BESS 1 is about to run out from $t=1.5h$ to $t=3h$, according to $E_1$, the SoC of BESS 1 has not converged to 0 and is still continuously decreasing. Anyway, SoC of all BESSs can ultimately converge to 0 simultaneously.}
	\subsection{Case 2. Comparison with Other Methods}
	\begin{figure}
	\centering
	\includegraphics[width=8cm]{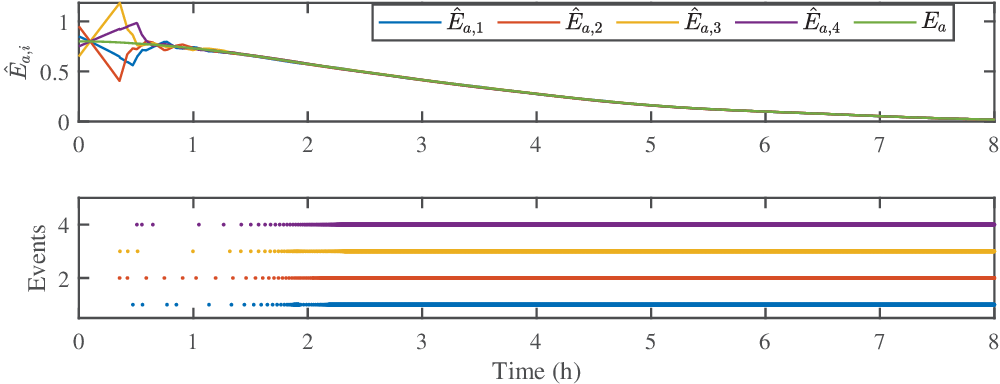}
	\caption{Evolution of SoC balance estimation and events under the scheme in \cite{Qian2023} in Case 2.}
	\label{case2a}
\end{figure}
\begin{figure}
	\centering
	\includegraphics[width=8cm]{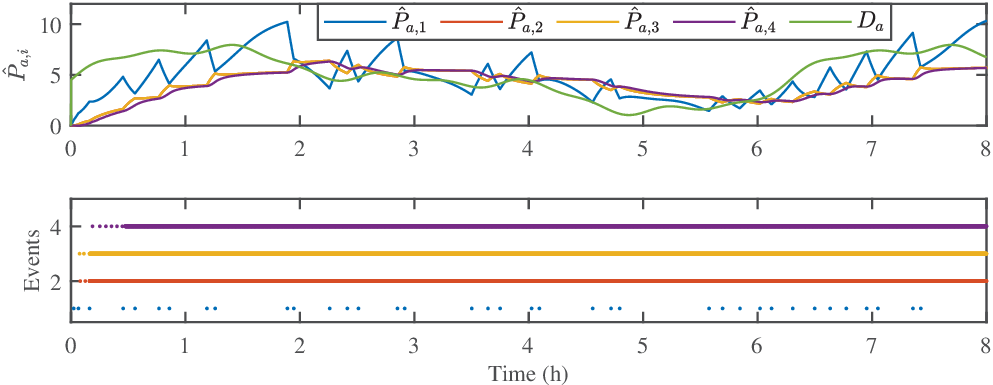}
	\caption{Evolution of power sharing estimation and events under the scheme in \cite{Qian2023} in Case 2.}
	\label{case2b}
\end{figure}
\begin{figure}
	\centering
	\includegraphics[width=8cm]{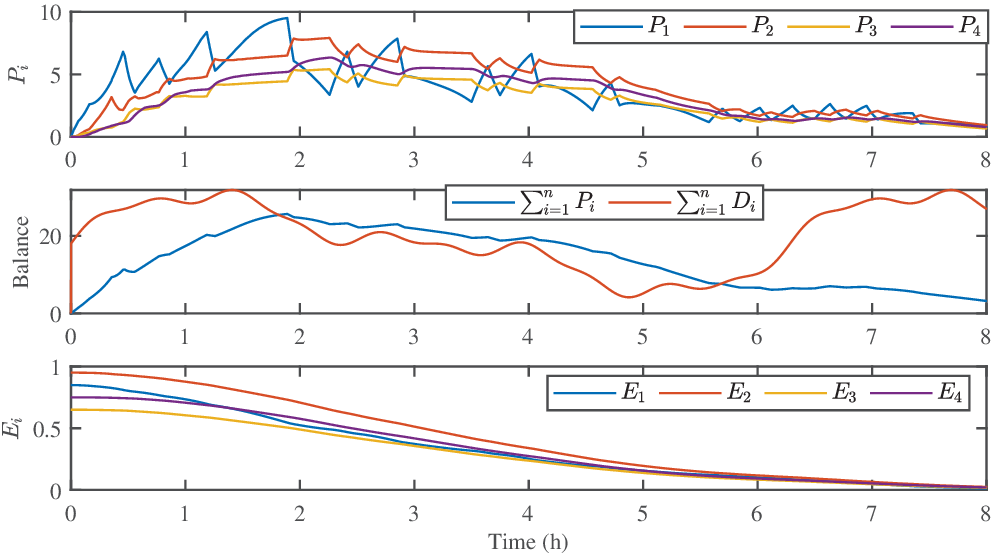}
	\caption{Evolution of output power and supply-demand balance under the scheme in \cite{Qian2023} in Case 2.}
	\label{case2c}
\end{figure}
	To illustrate the progressiveness, the effect of the distributed event-triggered estimators designed in \cite{Qian2023} \textcolor{blue}{is used for comparison}. From the perspective of the {power} sharing estimator, a centralized {method} needs to be designed to obtain the average load $P_a$ \textcolor{blue}{of} the entire network. Compared to this, the scheme adopted in this paper is more distributed. Besides, the average tracking performance and consensus performance of the scheme adopted are decoupled and regulated by different parameters. Specifically, increasing $\lambda$ and $\mathrm{k}$ leads to faster convergence speed, while reducing $\rho_\infty$ and increasing $\mathrm{k}_\mathrm{r}$ can improve control accuracy and average tracking performance, respectively.
	\par Then, the simulation results of this two estimators {designed} in \cite{Qian2023} are shown in \textcolor{blue}{Figs.} \ref{case2a}-\ref{case2c}. For \textcolor{blue}{the} SoC balance {state estimator}, the scheme proposed by the authors in \cite{Qian2023} performs better than the scheme designed in this paper {at the cost of triggering more and more frequently. Besides}, the {power} sharing estimator in \cite{Qian2023} do not perform as satisfactorily. To be frank, this estimator does not seem to have ideal tracking performance for rapidly changing loads, as there are some significant deviations in the tracking {and requires a faster triggering frequency}. In addition, there is always a significant deviation between supply and demand.
	\subsection{Case 3. The Test of {Play-and-Plug}}
		\begin{figure}
		\centering
		\includegraphics[width=8cm]{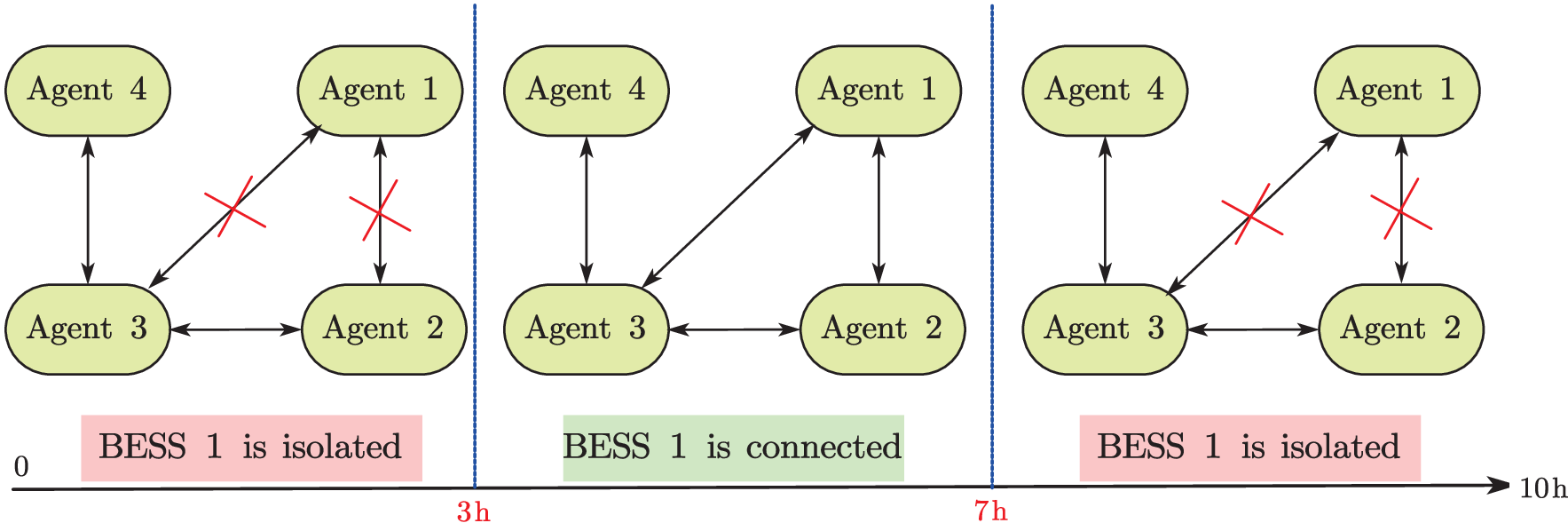}
		\caption{Simulation events in Case 3.}
		\label{case3}
	\end{figure}
	\begin{figure}
		\centering
		\includegraphics[width=8cm]{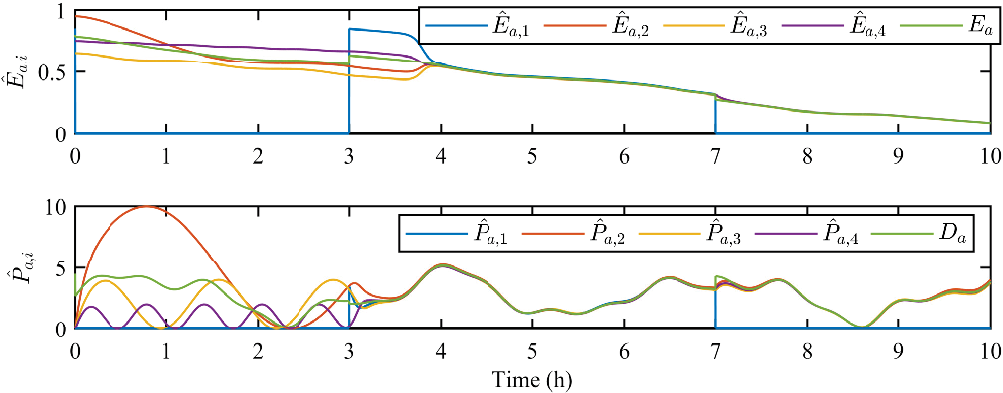}
			\caption{Evolution of SoC balance estimation and power sharing estimation in Case 3.}
		\label{case3a}
	\end{figure}
	\begin{figure}
		\centering
		\includegraphics[width=8cm]{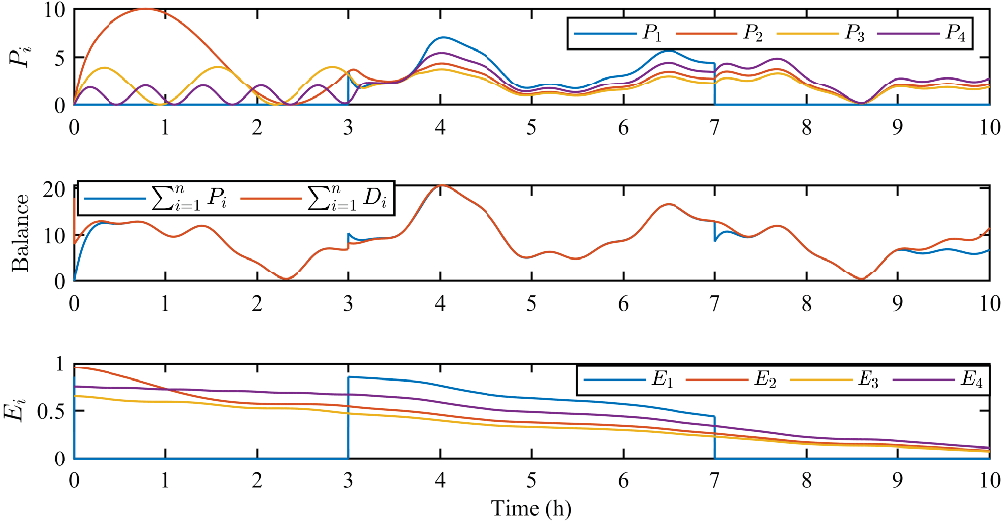}
		\caption{Evolution of output power, supply-demand balance, and SoC in Case 3.}
		\label{case3b}
	\end{figure}
	\par During normal operation of a BESS network, a bus may be isolated due to accidents or plans. Therefore, the robustness challenges brought about by communication failures should be included in the scheme design {to test the Play-and-Plug function}. In view of this, {the part where BESS 1 is located is removed at $t=0$ and $t=7h$, and restored at $t=3h$. Correspondingly, agent 1 encounters a communication failure at $t=0$ and $t=7h$, but resumes communication at $t=3h$. These events in this simulation can \textcolor{blue}{be referred} to in Fig. \ref{case3}.} The simulation results are shown in \textcolor{blue}{Figs.} \ref{case3a} and \ref{case3b}.
	\par When {agent} 1 encounters a communication failure, the remaining part remains connected, resulting in the situation shown in Fig. \ref{case3a}. Specifically, the SoC balance {state} and {power} sharing {state of the remaining part are} well estimated. Similarly, the remaining BESSs are ultimately almost simultaneously depleted, as shown in Fig. \ref{case3b}. Moreover, the supply-demand balance of the remaining {part} can still be well maintained. In summary, the designed scheme has good robustness {and Plug-and-Play function}.
	\subsection{Case 4 The Test on Scalability of the Designed Algorithm}
	\par {To test whether the designed algorithm allows for more BESSs, the number of BESSs in Fig. \ref{comm} is increased to 12. Under the designed distributed protocol \eqref{9a}-\eqref{9d} with event triggering mechanism \eqref{10}, the simulation results are shown in \textcolor{blue}{Figs.} \ref{case5a}-\ref{case5c}.}
	\par {From the simulation results, as shown in Fig. \ref{case5a}, all SoC balance estimation values can track the real-time SoC average value well, and the consensus error does not exceed the preset value. Similar conclusions \textcolor{blue}{are drawn on} power sharing estimators, as shown in Fig. \ref{case5b}. From Fig. \ref{case5c}, it can be seen that the supply-demand balance can also be well maintained, and the SoC of each BESS can ultimately converge to 0 simultaneously. It can be seen from here that the designed algorithm has a certain degree of scalability.}
	\begin{figure}
		\centering
		\includegraphics[width=8cm]{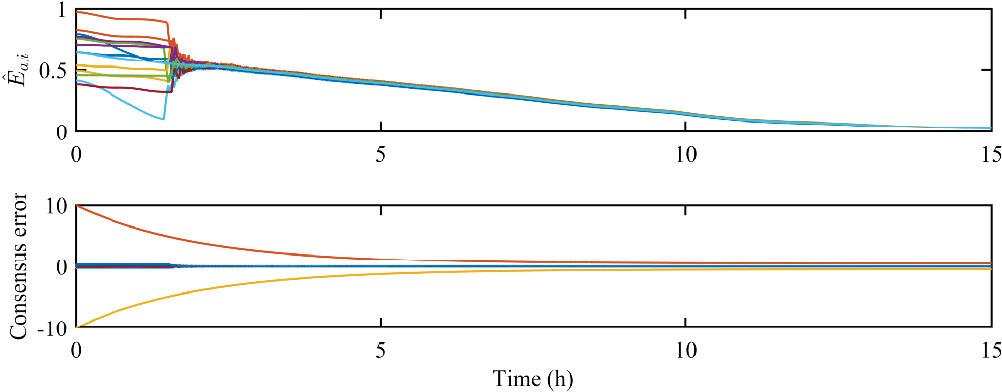}
		\caption{Evolution of SoC balance estimation and consensus error in Case 4.}
		\label{case5a}
	\end{figure}
	\begin{figure}
		\centering
		\includegraphics[width=8cm]{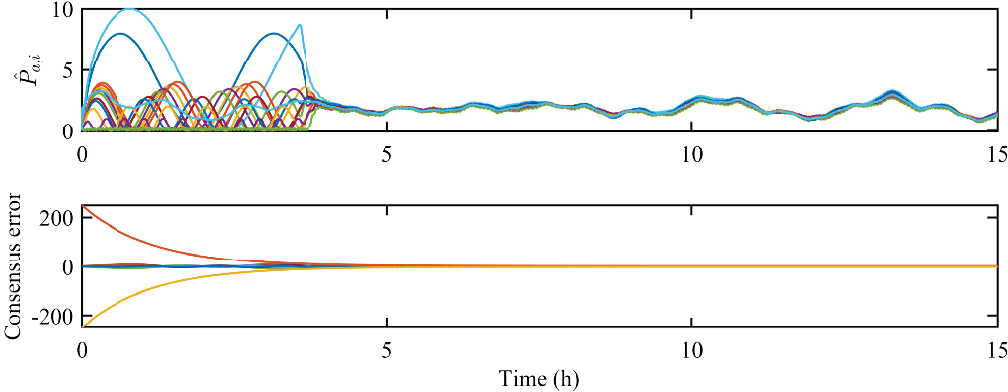}
		\caption{Evolution of power sharing estimation and consensus error in Case 4.}
		\label{case5b}
	\end{figure}
	\begin{figure}
		\centering
		\includegraphics[width=8cm]{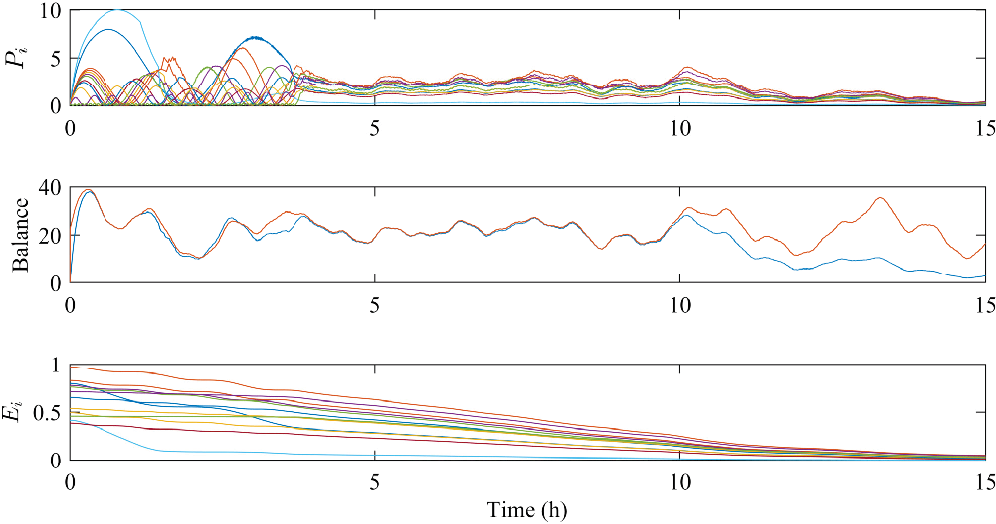}
		\caption{Evolution of output power, supply-demand balance, and SoC in Case 4.}
		\label{case5c}
	\end{figure}
		\subsection{Case 5. The Test with the Self-triggering Mechanism}
	\par To verify the effectiveness of the designed self triggering scheme, {a} simulation is conducted on \eqref{19} based on Case 1. The {simulation} results are shown in \textcolor{blue}{Figs.} \ref{case4a}-\ref{case4c}.
	\par From \textcolor{blue}{Figs.} \ref{case4a} and\ref{case4b}, it can be seen that each agent can still estimate the {power} sharing and SoC balance states with almost zero error. Meanwhile, the two errors can still converge to the desired range at a rate not less than the predefined {one}. Besides, compared to Case 1, both states no longer exhibit local oscillations. But the corresponding cost is to trigger more frequently, as described in Remark \ref{rem2}. \textcolor{blue}{From} Fig. \ref{case4c}, the supply-demand balance can still be well maintained{, and all BESSs are almost simultaneously depleted.} 
	\begin{figure}
		\centering
		\includegraphics[width=8cm]{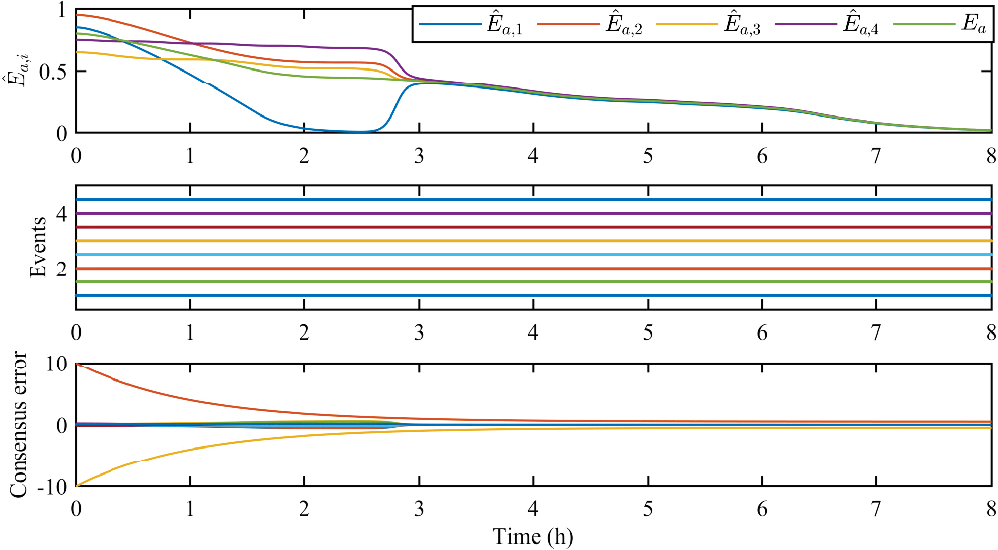}
		\caption{Evolution of SoC balance estimation, events, and consensus error in Case 5.}
		\label{case4a}
	\end{figure}
	\begin{figure}
		\centering
		\includegraphics[width=8cm]{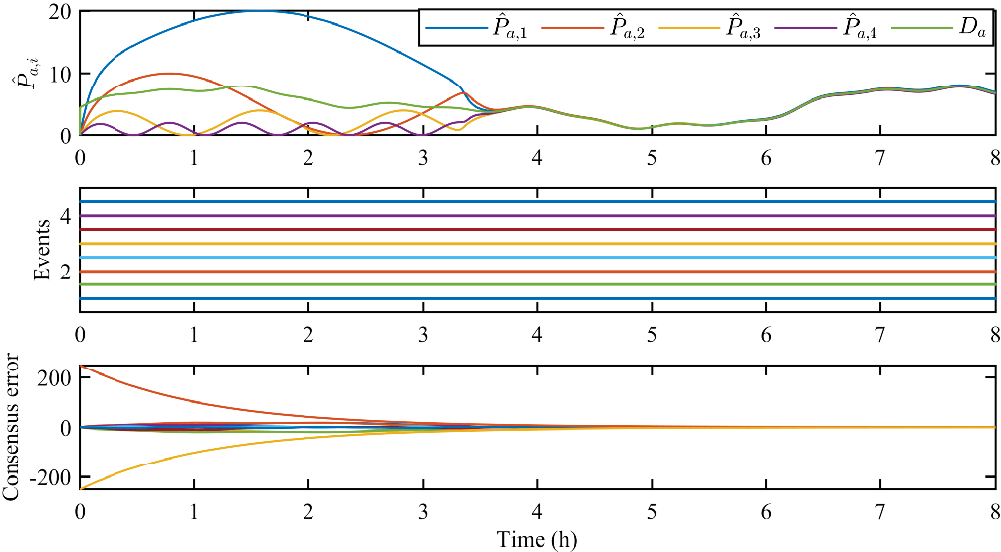}
		\caption{Evolution of power sharing estimation, events, and consensus error in Case 5.}
		\label{case4b}
	\end{figure}
	\begin{figure}
		\centering
		\includegraphics[width=8cm]{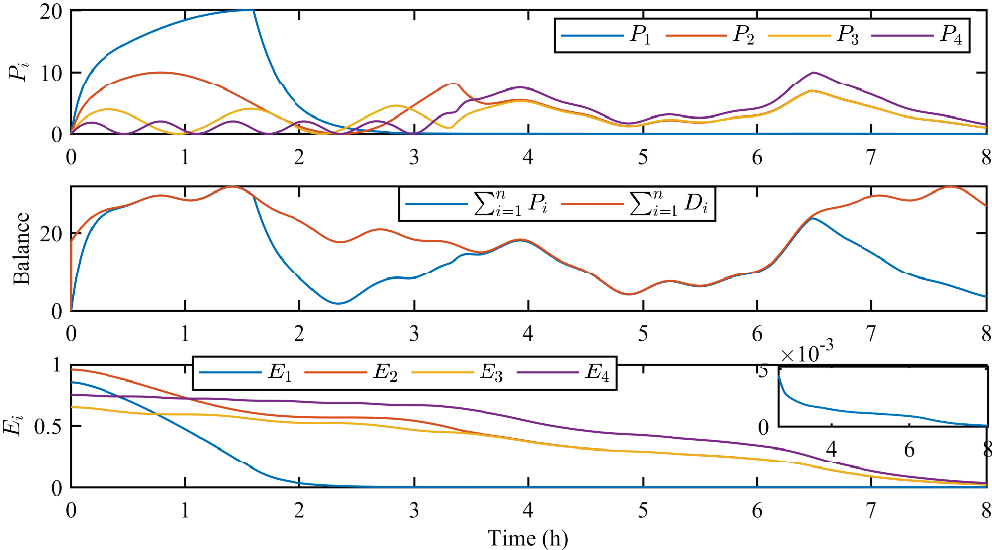}
		\caption{Evolution of output power, supply-demand balance, and SoC in Case 5.}
		\label{case4c}
	\end{figure}
	\section{Conclusions}
	\par In this paper, two distributed estimators with prescribed performance under intermittent {dynamics} are developed such that each agent managing a BESS grasps the current power sharing state and SoC balance state. This is done in order to ensure that any one of BESSs in a group with the same relative SoC change rate does not exit prematurely. The introduction of the PPC method completely decouples consensus and dynamic tracking performance, which is beneficial for the selection of gains. {An event triggering mechanism and a self triggering mechanism are designed to save actuators and communication resources.} The simulation results show that the designed scheme not only outperforms existing schemes in terms of consensus and dynamic tracking performance, but also has robustness and scalability, which are suitable for the needs of wide area microgrids. 
	\par {It should be pointed out that the designed self triggering scheme adopts the upper bound of consensus error, resulting in more frequent triggering. Perhaps a suitable observer is needed to observe the upper limit of consensus error in real time to improve the self triggering mechanism and reduce the triggering frequency.} \textcolor{blue}{In addition, the PPC method used is too common and can be somewhat challenging in dealing with situations such as actuator saturation and disturbance in multi-agent systems. Therefore, introducing a non-fragile PPC method \cite{9930639,BJF2022} into distributed DAT algorithms will be a promising and urgent problem to be solved.}
	\appendices
	\ifCLASSOPTIONcaptionsoff
	\newpage
	\fi
	
	\bibliographystyle{IEEEtran}
	\bibliography{re0.bib}

@article{HossainLipu2022,
	title = {A Review of Controllers and Optimizations Based Scheduling Operation for Battery Energy Storage System towards Decarbonization in Microgrid: Challenges and Future Directions},
	shorttitle = {A Review of Controllers and Optimizations Based Scheduling Operation for Battery Energy Storage System towards Decarbonization in Microgrid},
	author = {Hossain Lipu, M. S. and Ansari, Shaheer and Miah, Md. Sazal and Hasan, Kamrul and Meraj, Sheikh T. and Faisal, M. and Jamal, Taskin and Ali, Sawal H. M. and Hussain, Aini and Muttaqi, Kashem M. and Hannan, M. A.},
	year = {2022},
	journal = {Journal of Cleaner Production},
	shortjournal = {Journal of Cleaner Production},
	volume = {360},
	pages = {132188},
	issn = {0959-6526},
	doi = {10.1016/j.jclepro.2022.132188}
}

@article{ZHAO2014538,
	title = {Review on the costs and benefits of renewable energy power subsidy in China},
	author = {Zhao, Huiru and Guo, Sen and Fu, Liwen},
	journal = {Renewable and Sustainable Energy Reviews},
	volume = {37},
	pages = {538-549},
	year = {2014},
	issn = {1364-0321},
	doi = {https://doi.org/10.1016/j.rser.2014.05.061}
}

@article{solyaliComprehensiveStateoftheartReview2022,
	title = {A Comprehensive State-of-the-Art Review of Electrochemical Battery Storage Systems for Power Grids},
	author = {Solyali, Davut and Safaei, Babak and Zargar, Omid and Aytac, Gizem},
	year = {2022},
	journal = {International Journal of Energy Research},
	volume = {46},
	number = {13},
	pages = {17786--17812},
	issn = {1099-114X},
	doi = {10.1002/er.8451},
	urldate = {2023-03-13},
	langid = {english}
}

@article{Zeng2022l,
	title = {Hierarchical Cooperative Control Strategy for Battery Storage System in Islanded DC Microgrid},
	author = {Zeng, Yuji and Zhang, Qinjin and Liu, Yancheng and Zhuang, Xuzhou and Guo, Haohao},
	year = {2022},
	journal = {IEEE Transactions on Power Systems},
	volume = {37},
	number = {5},
	pages = {4028--4039},
	issn = {1558-0679},
	doi = {10.1109/TPWRS.2021.3131591}
}

@article{Wang2020d,
	title = {Consensus-Based Control of Hybrid Energy Storage System With a Cascaded Multiport Converter in DC Microgrids},
	author = {Wang, Benfei and Wang, Yu and Xu, Yan and Zhang, Xinan and Gooi, Hoay Beng and Ukil, Abhisek and Tan, Xiaojun},
	YEAR = {2020},
	journal = {IEEE Transactions on Sustainable Energy},
	volume = {11},
	number = {4},
	pages = {2356--2366},
	issn = {1949-3037},
	doi = {10.1109/TSTE.2019.2956054}
}

@article{Litjens2018,
	title = {Economic Benefits of Combining Self-Consumption Enhancement with Frequency Restoration Reserves Provision by Photovoltaic-Battery Systems},
	author = {Litjens, G. B. M. A. and Worrell, E. and van Sark, W. G. J. H. M.},
	year = {2018},
	journal = {Applied Energy},
	shortjournal = {Appl. Energy},
	volume = {223},
	pages = {172--187},
	publisher = {Elsevier Sci Ltd},
	location = {Oxford},
	issn = {0306-2619},
	doi = {10.1016/j.apenergy.2018.04.018}
}

@article{Yang2022,
	title = {Modelling and Optimal Energy Management for Battery Energy Storage Systems in Renewable Energy Systems: A Review},
	shorttitle = {Modelling and Optimal Energy Management for Battery Energy Storage Systems in Renewable Energy Systems},
	author = {Yang, Yuqing and Bremner, Stephen and Menictas, Chris and Kay, Merlinde},
	year = {2022},
	journal = {Renewable and Sustainable Energy Reviews},
	shortjournal = {Renewable and Sustainable Energy Reviews},
	volume = {167},
	pages = {112671},
	issn = {1364-0321},
	doi = {10.1016/j.rser.2022.112671}
}

@inproceedings{Shinichi2017,
	title = {Demonstration Project of Power System Stabilization with the Hybrid Battery Energy Storage System},
	booktitle = {2017 IEEE International Telecommunications Energy Conference (Intelec)},
	author = {Shinichi, Sano and Shimoura, Ichiro},
	year = {2017},
	pages = {262--269},
	location = {New York},
	issn = {0275-0473}
}

@inproceedings{Sokol2017,
	title = {The Specificity of Electrical Energy Storage Unit Application},
	booktitle = {2017 IEEE First Ukraine Conference on Electrical and Computer Engineering (Ukrcon)},
	author = {Sokol, Evgen and Zamaruiev, Volodymyr and Kryvosheev, Serhiy and Styslo, Bohdan and Makarov, Vadym},
	pages = {432--435},
	year = {2017},
	location = {New York}
}

@article{Baveja2023,
	title = {Predicting Temperature Distribution of Passively Balanced Battery Module under Realistic Driving Conditions through Coupled Equivalent Circuit Method and Lumped Heat Dissipation Method},
	author = {Baveja, Rajveer and Bhattacharya, Jishnu and Panchal, Satyam and Fraser, Roydon and Fowler, Michael},
	year = {2023},
	journal = {Journal of Energy Storage},
	shortjournal = {Journal of Energy Storage},
	volume = {70},
	pages = {107967},
	issn = {2352-152X},
	doi = {10.1016/j.est.2023.107967}
}

@article{Cao2021,
	title = {Hierarchical SOC Balancing Controller for Battery Energy Storage System},
	author = {Cao, Yuan and Abu Qahouq, Jaber A.},
	year = {2021},
	journal = {IEEE Transactions on Industrial Electronics},
	volume = {68},
	number = {10},
	pages = {9386--9397},
	issn = {1557-9948},
	doi = {10.1109/TIE.2020.3021608}
}

@article{Cai2016,
	title = {Distributed Control Scheme for Package-Level State-of-Charge Balancing of Grid-Connected Battery Energy Storage System},
	author = {Cai, He and Hu, Guoqiang},
	year = {2016},
	journal = {IEEE Transactions on Industrial Informatics},
	volume = {12},
	number = {5},
	pages = {1919--1929},
	issn = {1941-0050},
	doi = {10.1109/TII.2016.2601904}
}

@article{Khazaei2019a,
	title = {Multi-Agent Consensus Design for Heterogeneous Energy Storage Devices With Droop Control in Smart Grids},
	author = {Khazaei, Javad and Nguyen, Dinh Hoa},
	year = {2019},
	journal = {IEEE Transactions on Smart Grid},
	volume = {10},
	number = {2},
	pages = {1395--1404},
	issn = {1949-3061},
	doi = {10.1109/TSG.2017.2765241}
}

@article{Ding2020,
	title = {Distributed Resilient Finite-Time Secondary Control for Heterogeneous Battery Energy Storage Systems Under Denial-of-Service Attacks},
	author = {Ding, Lei and Han, Qing-Long and Ning, Boda and Yue, Dong},
	year = {2020},
	journal = {IEEE Transactions on Industrial Informatics},
	volume = {16},
	number = {7},
	pages = {4909--4919},
	issn = {1941-0050},
	doi = {10.1109/TII.2019.2955739}
}

@article{Hu2019,
	title = {Distributed Finite-Time Consensus Control for Heterogeneous Battery Energy Storage Systems in Droop-Controlled Microgrids},
	author = {Hu, Junyan and Lanzon, Alexander},
	year = {2019},
	journal = {IEEE Transactions on Smart Grid},
	volume = {10},
	number = {5},
	pages = {4751--4761},
	issn = {1949-3061},
	doi = {10.1109/TSG.2018.2868112}
}

@article{Nguyen2021a,
	title = {Unified Distributed Control of Battery Storage With Various Primary Control in Power Systems},
	author = {Nguyen, Dinh Hoa and Khazaei, Javad},
	year = {2021},
	journal = {IEEE Transactions on Sustainable Energy},
	volume = {12},
	number = {4},
	pages = {2332--2341},
	issn = {1949-3037},
	doi = {10.1109/TSTE.2021.3091976}
}

@article{Zhang2020,
	title = {A Distributed Step-by-Step Finite-Time Consensus Design for Heterogeneous Battery Energy Storage Devices with Droop Control},
	author = {Zhang, Yalin and Song, Yunzhong and Fei, Shumin},
	YEAR = {2023},
	month=sep,
	volume = {9},
	number = {5},
	journal = {CSEE Journal of Power and Energy Systems},
	pages = {1893-1904},
	issn = {2096-0042},
	doi = {10.17775/CSEEJPES.2020.00030}
}

@article{Hong2021d,
	title = {A Novel Multi-Agent Model-Free Control for State-of-Charge Balancing Between Distributed Battery Energy Storage Systems},
	author = {Hong, Yujin and Xu, Dezhi and Yang, Weilin and Jiang, Bin and Yan, Xing-Gang},
	year = {2021},
	journal = {IEEE Transactions on Emerging Topics in Computational Intelligence},
	shortjournal = {IEEE Trans. Emerg. Top. Comput. Intell.},
	volume = {5},
	number = {4},
	pages = {679--688},
	publisher = {Ieee-Inst Electrical Electronics Engineers Inc},
	location = {Piscataway},
	issn = {2471-285X},
	doi = {10.1109/TETCI.2020.2978434}
}

@article{Zhang2019e,
	title = {Nonlinear Sliding Mode and Distributed Control of Battery Energy Storage and Photovoltaic Systems in AC Microgrids With Communication Delays},
	author = {Zhang, Runfan and Hredzak, Branislav},
	YEAR = {2019},
	journal = {IEEE Transactions on Industrial Informatics},
	volume = {15},
	number = {9},
	pages = {5149--5160},
	issn = {1941-0050},
	doi = {10.1109/TII.2019.2896032}
}

@article{Xing2019,
	title = {Distributed State-of-Charge Balance Control With Event-Triggered Signal Transmissions for Multiple Energy Storage Systems in Smart Grid},
	author = {Xing, Lantao and Mishra, Yateendra and Tian, Yu-Chu and Ledwich, Gerard and Zhou, Chunjie and Du, Wenli and Qian, Feng},
	year = {2019},
	journal = {IEEE Transactions on Systems, Man, and Cybernetics: Systems},
	volume = {49},
	number = {8},
	pages = {1601--1611},
	issn = {2168-2232},
	doi = {10.1109/TSMC.2019.2916152}
}

@article{Meng2021,
	title = {Distributed Cooperative Control of Battery Energy Storage Systems in DC Microgrids},
	author = {Meng, Tingyang and Lin, Zongli and Shamash, Yacov A.},
	year = {2021},
	journal = {IEEE/CAA Journal of Automatica Sinica},
	volume = {8},
	number = {3},
	pages = {606--616},
	issn = {2329-9274},
	doi = {10.1109/JAS.2021.1003874}
}

@article{Meng2022,
	title = {State-of-Charge Balancing for Battery Energy Storage Systems in DC Microgrids by Distributed Adaptive Power Distribution},
	author = {Meng, Tingyang and Lin, Zongli and Wan, Yan and Shamash, Yacov A.},
	year = {2022},
	journal = {IEEE Control Systems Letters},
	volume = {6},
	pages = {512--517},
	issn = {2475-1456},
	doi = {10.1109/LCSYS.2021.3082103}
}

@article{Qian2023,
	title = {Distributed event-triggered algorithms for the management of networked battery systems},
	author = {Qian, Yangyang and Meng, Tingyang and Lin, Zongli and Wan, Yan and Shamash, Yacov A.},
	year = {2023},
	journal = {International Journal of Robust and Nonlinear Control},
	volume = {},
	pages = {1--26}
}

@article{Wu2022,
	title = {Distributed Multirate Control of Battery Energy Storage Systems for Power Allocation},
	author = {Wu, Han and Chai, Li and Tian, Yu-Chu},
	year = {2022},
	journal = {IEEE Transactions on Industrial Informatics},
	volume = {18},
	number = {12},
	pages = {8745--8754},
	issn = {1941-0050},
	doi = {10.1109/TII.2022.3153055}
}

@article{Bechlioulis2008,
	title = {Robust Adaptive Control of Feedback Linearizable MIMO Nonlinear Systems With Prescribed Performance},
	author = {Bechlioulis, Charalampos P. and Rovithakis, George A.},
	year = {2008},
	journal = {IEEE Transactions on Automatic Control},
	volume = {53},
	number = {9},
	pages = {2090--2099},
	issn = {1558-2523},
	doi = {10.1109/TAC.2008.929402}
}

@article{Bechlioulis2009,
	title = {Adaptive Control with Guaranteed Transient and Steady State Tracking Error Bounds for Strict Feedback Systems},
	author = {Bechlioulis, Charalampos P. and \vspace{0mm} Rovithakis, George A.},
	year = {2009},	
	journal = {Automatica},	
	shortjournal = {Automatica},	
	volume = {45},
	number = {2},	
	pages = {532--538},	
	issn = {0005-1098},	
	doi = {10.1016/j.automatica.2008.08.012}
}

@article{Bechlioulis2010,
	title = {Prescribed Performance Adaptive Control for Multi-Input Multi-Output Affine in the Control Nonlinear Systems},
	author = {Bechlioulis, Charalampos P. and \vspace{0mm} \vspace{0mm} Rovithakis, George A.},
	year = {2010},
	journal = {IEEE Transactions on Automatic Control},
	volume = {55},
	number = {5},
	pages = {1220--1226},
	issn = {1558-2523},
	doi = {10.1109/TAC.2010.2042508}
}

@ARTICLE{10101860,
	author={Lv, Maolong and Wang, Ning},
	journal={IEEE Transactions on Automatic Control}, 
	title={Distributed Control for Uncertain Multiagent Systems With the Powers of Positive-Odd Numbers: A Low-Complexity Design Approach}, 
	year={2024},
	month=jan,
	volume={69},
	number={1},
	pages={434-441},
	doi={10.1109/TAC.2023.3266986}
}

@ARTICLE{10198705,
	author={Ma, Chi and Dong, Dianbiao},
	journal={IEEE/CAA Journal of Automatica Sinica}, 
	title={Finite-Time Prescribed Performance Time-Varying Formation Control for Second-Order Multi-Agent Systems with Non-Strict Feedback Based on a Neural Network Observer}, 
	year={2024},
	month=apr,
	volume={11},
	number={4},
	pages={1039-1050},
	doi={10.1109/JAS.2023.123615}
}

@ARTICLE{10486878,
	author={Liu, Donghao and Mao, Zehui and Jiang, Bin and Yan, Xing-Gang},
	journal={IEEE Transactions on Cybernetics}, 
	title={Prescribed Performance Fault-Tolerant Control for Synchronization of Heterogeneous Nonlinear MASs Using Reinforcement Learning}, 
	year={2024},
	month=sep,
	volume={54},
	number={9},
	pages={5451-5462},
	doi={10.1109/TCYB.2024.3374349}
}

@ARTICLE{10428061,
	author={Hu, Wenfeng and Hou, Yahui and Chen, Zhiyong and Yang, Chunhua and Gui, Weihua},
	journal={IEEE Transactions on Automatic Control}, 
	title={Event-Triggered Consensus of Multiagent Systems With Prescribed Performance}, 
	year={2024},
	month=aug,
	volume={69},
	number={8},
	pages={5462-5469},
	doi={10.1109/TAC.2024.3364017}
}

@ARTICLE{10578027,
	author={Hou, Yahui and Hu, Wenfeng and Li, Jianqi and Huang, Tingwen},
	journal={IEEE Transactions on Circuits and Systems I: Regular Papers}, 
	title={Prescribed Performance Control for Double-Integrator Multi-Agent Systems: A Unified Event-Triggered Consensus Framework}, 
	year={2024},
	month=sep,
	volume={71},
	number={9},
	pages={4222-4232},
	doi={10.1109/TCSI.2024.3416400}
}

@ARTICLE{10418189,
	author={Bi, Wenshan and Zhang, Chen and Sui, Shuai and Tong, Shaocheng and Chen, C. L. Philip},
	journal={IEEE Transactions on Fuzzy Systems}, 
	title={Fixed-Time Fuzzy Adaptive Bipartite Output Consensus Tracking for Nonlinear Coopetition MASs}, 
	year={2024},
	month=may,
	volume={32},
	number={5},
	pages={2775-2785},
	doi={10.1109/TFUZZ.2024.3360283}
}

@ARTICLE{10398255,
	author={Yang, Tingting and Dong, Jiuxiang},
	journal={IEEE Transactions on Systems, Man, and Cybernetics: Systems}, 
	title={Distributed Event-Triggered Fixed-Time DSC of Multiagent Systems}, 
	year={2024},
	month=apr,
	volume={54},
	number={4},
	pages={2484-2494},
	doi={10.1109/TSMC.2023.3344269}
}

@ARTICLE{10711307,
	author={Zhang, Dun and Lam, James and Xie, Xiaochen and Fan, Chenchen and Song, Xiaoqi},
	journal={IEEE Transactions on Cybernetics}, 
	title={Fault-Tolerant Consensus of Multiagent Systems With Prescribed Performance}, 
	year={2024},
	volume={},
	number={},
	pages={1-14},
	doi={10.1109/TCYB.2024.3467217}
}

@book{Sontag1998,
	title = {Mathematical Control Theory},
	author = {Sontag, Eduardo D.},
	editorb = {Marsden, J. E. and Sirovich, L. and Golubitsky, M. and Jäger, W.},
	editorbtype = {redactor},
	year = {1998},
	publisher = {Springer},
	location = {New York, NY},
	doi = {10.1007/978-1-4612-0577-7}
}

@ARTICLE{9770469,
	author={Zhang, Juan and Zhang, Huaguang and Sun, Shaoxin and Cai, Yuliang},
	journal={IEEE Transactions on Cybernetics}, 
	title={Adaptive Time-Varying Formation Tracking Control for Multiagent Systems With Nonzero Leader Input by Intermittent Communications}, 
	year={2023},
	month=sep,
	volume={53},
	number={9},
	pages={5706-5715},
	doi={10.1109/TCYB.2022.3165212}
}

@ARTICLE{9930639,
	author={Bu, Xiangwei and Jiang, Baoxu and Lei, Humin},
	journal={IEEE Transactions on Fuzzy Systems}, 
	title={Low-Complexity Fuzzy Neural Control of Constrained Waverider Vehicles via Fragility-Free Prescribed Performance Approach}, 
	year={2023},
	month=jul,
	volume={31},
	number={7},
	pages={2127-2139},
	doi={10.1109/TFUZZ.2022.3217378}
}

@ARTICLE{BJF2022,
	author={Bu, Xiangwei and Jiang, Baoxu and Feng, Yin'an},
	journal={Nonlinear Dynamics}, 
	title={Non-fragile tracking control of constrained Waverider Vehicles with readjusting prescribed performance}, 
	year={2022},
	month=apr,
	volume={108},
	pages={3657–3669},
	doi={10.1007/s11071-022-07430-6}
}
	\vspace{-15 mm} 
\end{document}